\documentclass{article}

\PassOptionsToPackage{numbers, compress}{natbib}

\usepackage[preprint]{neurips_2026}

\usepackage[utf8]{inputenc} 
\usepackage[T1]{fontenc}    
\usepackage{hyperref}       
\usepackage{url}            
\usepackage{booktabs}       
\usepackage{amsfonts}       
\usepackage{nicefrac}       
\usepackage{microtype}      
\usepackage{xcolor}         
\usepackage{enumitem}
\usepackage{amsmath,amssymb,amsthm}
\usepackage{algorithm}
\usepackage{algpseudocode}

\newtheorem{theorem}{Theorem}[section]
\newtheorem{lemma}{Lemma}[section]
\newtheorem{proposition}{Proposition}[section]
\newtheorem{remark}{Remark}[section]
\newtheorem{definition}{Definition}[section] 

\usepackage{graphicx} 
\providecommand{\keywords}[1]{\par\medskip\noindent\textbf{Keywords: } #1}

\title{Deep Learning under Fractional-Order Differential Privacy}

\author{%
  Mohammad Partohaghighi \\
  MESA Lab, Dept. of EECS\\
  University of California, Merced\\
  Merced, CA 95343 \\
  \texttt{mpartohaghighi@ucmerced.edu} \\
  \And
  Roummel Marcia \\
  Dept. of Applied Mathematics\\
  University of California, Merced\\
  Merced, CA 95343 \\
  \texttt{rmarcia@ucmerced.edu} \\
}

\begin{document}

\maketitle

\maketitle

\begin{abstract}
Differentially private stochastic gradient descent (DP-SGD) is a standard
approach to privacy-preserving learning based on per-example gradient clipping,
random subsampling, Gaussian perturbation, and privacy accounting. In classical
DP-SGD, each private release depends only on the current clipped subsampled
gradient sum. We propose Fractional-Order Differentially Private Stochastic
Gradient Descent (\textbf{FO-DP-SGD}), a mechanism-level extension of DP-SGD in which,
before Gaussian perturbation, the standard clipped subsampled sum query is
replaced by a fractional-order recursive query that combines the current clipped
sum with a finite-window, power-law-weighted aggregation of previously released
private sum-level outputs. This introduces fractional memory directly into the
release mechanism while preserving the standard \emph{sum-then-noise-then-divide}
structure.

Under add/remove adjacency with Poisson subsampling, we show that in the
current-step sensitivity analysis, the only newly data-dependent term is the
scaled current clipped sum, yielding a history-conditioned effective
\(\ell_2\)-sensitivity bound of at most \(\beta C\), where \(C\) is the
clipping threshold and \(\beta\in(0,1]\) controls the current-step contribution.
As a result, FO-DP-SGD admits standard per-step R\'enyi differential privacy
accounting via a Poisson-subsampled Gaussian mechanism with effective
noise-to-sensitivity ratio \(\sigma/\beta\), and composes naturally to provide
overall \((\varepsilon,\delta)\)-differential privacy guarantees.

Beyond privacy preservation, FO-DP-SGD provides a principled mechanism-level
framework for studying long-memory effects in private optimization. Our analysis
clarifies how the fractional order, memory window, and mixing coefficient govern
the trade-off among current-step sensitivity, direct signal retention, and
inherited private-history influence. Experiments on SVHN, CIFAR-10, and
CIFAR-100 show that FO-DP-SGD improves test accuracy and privacy--utility
performance relative to DP-SGD and several private optimization baselines,
including DP-Adam, DP-IS, SA-DP-SGD, ADP-AdamW, DP-SAT, and DP-Adam-AC.
\end{abstract}

\keywords{
Differential Privacy;
DP-SGD;
Fractional Calculus

}


\section{Introduction}
\label{sec:introduction}

Modern machine learning systems are increasingly trained on data that are
intrinsically sensitive, including medical records, behavioral traces,
financial histories, and other person-level information whose misuse may create
legal, ethical, and societal harms. Differential privacy (DP) has emerged as a
leading formal framework for addressing these risks by ensuring that the output
distribution of a computation does not depend strongly on the participation of
any single individual \cite{DworkMNS06,DworkRoth14,KasiviswanathanLNRS11}. In
learning settings, this challenge is especially acute because privacy must be
maintained not for a single statistic, but across an iterative optimization
procedure whose intermediate computations repeatedly access sensitive data
\cite{ChaudhuriMS11,BassilyST14}. The core problem is therefore to design
training procedures that preserve meaningful privacy guarantees while retaining
useful optimization and statistical performance.

Among practical approaches to private training, differentially private
stochastic gradient descent (DP-SGD) has become the canonical baseline
\cite{AbadiCGMMTZ16,SongCS13}. At each iteration, per-example gradients are
clipped to control sensitivity, aggregated over a subsample, perturbed with
Gaussian noise, and then used for parameter updating. Combined with privacy
amplification by subsampling and suitable composition analysis, this simple
clipped-sum Gaussian mechanism has enabled private training of deep models
across many tasks \cite{AbadiCGMMTZ16,PapernotSMR16,PapernotSESMR18,McMahanRTZ18}.
Precisely because of its broad adoption, DP-SGD is the natural reference point
for new work in private optimization.

At the same time, the optimization limitations of DP-SGD are well known.
Gradient clipping introduces bias, while Gaussian perturbation adds variance
that accumulates across training steps \cite{AbadiCGMMTZ16,AndrewKMSV21,AsiDFJT21}.
In deep and nonconvex settings, these effects can slow convergence, reduce
final utility, and make performance highly sensitive to clipping norms, noise
scales, subsampling rates, and step sizes \cite{KairouzMSTTX21,LiZLRMS23}. A
large body of prior work has therefore improved private learning through tighter
privacy accounting, optimizer-side refinements, architecture design, or
alternative training structures. However, in many such methods, the core
private release itself remains unchanged: the clipped subsampled sum is still
privatized in the standard way, and improvements occur mainly in how that
released quantity is subsequently used.

This leaves open a more direct question: can the private query itself be made
more informative while preserving compatibility with the standard
subsampled-Gaussian accounting framework? This paper answers that question in
the affirmative. We propose \emph{Fractional-Order Differentially Private
Stochastic Gradient Descent} (\textbf{FO-DP-SGD}), a mechanism-level extension of
DP-SGD in which the usual clipped subsampled sum query is replaced, before
Gaussian perturbation, by a recursive query with structured private-history
dependence.

Our design is motivated by fractional calculus, which provides a principled
framework for modeling nonlocal dependence, long memory, and power-law temporal
effects \cite{OldhamS74,MillerR93,SamkoKM93,Podlubny98,KilbasST06,Caputo67}.
Fractional-order methods have 
been used in optimization to enrich
iterative dynamics beyond standard first-order recursions
\cite{WeiKYW18,BaoPZ18,ShengWCW20}. We bring this perspective to private
learning at the mechanism level 
not
only at the optimizer level. In
FO-DP-SGD, the current clipped subsampled sum enters the private query with
coefficient $\beta\in(0,1]$, while previously released \emph{private}
sum-level outputs are incorporated through a normalized fractional memory
kernel. To make this history selective rather than purely lag-based, we form a
causal exponential moving average (EMA) trend from prior private releases,
measure the inconsistency of each lagged release relative to that trend, and
temper its contribution accordingly. A confidence factor derived from the EMA
magnitude further moderates the tempering when the private trend is weak or
noisy.

This viewpoint is both conceptual and technical. FO-DP-SGD is not merely
DP-SGD followed by a history-aware optimizer; rather, the release mechanism
itself is recursive and transcript-dependent before normalization and parameter
updating. This distinction matters because privacy must be established for the
released query itself, not justified only through post-processing intuition
\cite{DworkRoth14,Mironov17,WangBK19}. The key privacy observation is that,
conditioned on the realized Poisson subsampling mask and the prior private
transcript, the EMA trend, inconsistency scores, confidence factor, and
normalized memory weights are fixed. Hence, the only newly data-dependent term
in the current recursive query is the scaled clipped-sum term, yielding an
effective $\ell_2$-sensitivity bound of at most $\beta C$, where $C$ is the
clipping norm. As a result, FO-DP-SGD remains compatible with the
Poisson-subsampled Gaussian mechanism and standard R\'enyi differential privacy
(RDP) accounting under add/remove adjacency
\cite{Mironov17,WangBK19,BalleW18,DongRS22}.

This mechanism-level modification creates a new design axis for private
optimization. Instead of treating the clipped-sum Gaussian release as a fixed
primitive and modifying only the downstream optimizer, FO-DP-SGD enriches the
private release itself with selectively trusted history while preserving the
classical sum-then-noise-then-divide structure. In this way, the method opens a
new privacy--utility trade-off between fully memoryless DP-SGD and purely
optimizer-side refinements.

Our main contributions are as follows:
\begin{itemize}[leftmargin=*]
\item We introduce \emph{FO-DP-SGD}, a mechanism-level extension of DP-SGD in
which the private release combines the current clipped subsampled sum with a
fractional private memory of past sum-level releases.
\item We augment this private memory with an EMA-based private trend,
inconsistency-aware tempering, and confidence-aware scaling, yielding a causal
and transcript-dependent release mechanism that selectively suppresses stale or
unreliable history.
\item We show that the proposed method preserves the classical
sum-then-noise-then-divide structure of DP-SGD and reduces to standard DP-SGD
when $\beta=1$, making it a strict generalization of the canonical baseline.
\item We provide a privacy analysis under Poisson subsampling, add/remove
adjacency, and RDP accounting, and show that the history- and
mask-conditioned sensitivity of the recursive query remains bounded by
$\beta C$.
\item We empirically evaluate FO-DP-SGD on SVHN, CIFAR-10, and CIFAR-100 and
show that it improves privacy--utility performance and convergence behavior
relative to standard DP-SGD and relevant private optimization baselines.
\end{itemize}

\paragraph{The remainder of this work.}
Section~\ref{sec:related_work} reviews related work on private optimization,
privacy accounting, and history-aware optimization. 
Section~\ref{sec:fodpsgd_methodology} introduces the FO-DP-SGD mechanism,
including the recursive private query, fractional memory kernel, and
confidence-aware tempering rule. 
Section~\ref{sec:fodpsgd_theory} provides the theoretical analysis, including
history-conditioned sensitivity, privacy accounting, and limiting regimes. 
Section~\ref{sec:experiment} presents the experimental setup, reproducibility
details, comparisons across SVHN, CIFAR-10, and CIFAR-100, and ablation studies.
Section~\ref{sec:conclusion} concludes the paper and summarizes limitations and
future directions. 
Section~\ref{sec:broader-impact} discusses broader impacts and deployment
considerations.


\section{Related Work}
\label{sec:related_work}

Research on privacy-preserving learning spans statistical, algorithmic, and
systems-oriented directions. At its foundation, classical differential privacy
established the formal framework for reasoning about individual-level
protection under randomized computation
\cite{DworkMNS06,DworkRoth14,KasiviswanathanLNRS11}. In machine learning, this
perspective evolved into private empirical risk minimization and private
stochastic optimization, where privacy loss must be controlled across
iterative, data-dependent updates
\cite{ChaudhuriMS11,BassilyST14,BassilyFTT19}.

Among practical methods, differentially private stochastic gradient descent
(DP-SGD) has become the dominant baseline for private deep learning because of
its simple and modular mechanism: per-example gradient clipping controls
sensitivity, random subsampling provides privacy amplification, and Gaussian
perturbation supports composition through modern privacy-accounting frameworks
\cite{AbadiCGMMTZ16,SongCS13}. Its limitations are also well known: clipping
distorts the optimization signal, injected noise accumulates over many
iterations, and performance can be highly sensitive to the clipping norm, noise
multiplier, sampling rate, and learning rate
\cite{AbadiCGMMTZ16,AndrewKMSV21,AsiDFJT21,KairouzMSTTX21,LiZLRMS23}.

A major line of work has therefore focused on improving privacy accounting for
iterative Gaussian mechanisms. Moments-based accounting was later generalized
and refined through R\'enyi differential privacy, subsampled-RDP analyses, and
related composition tools, enabling tighter and more flexible end-to-end
privacy guarantees for private training
\cite{AbadiCGMMTZ16,Mironov17,WangBK19,BunS16,DongRS22,FeldmanZ21}. These works
are complementary to our approach: FO-DP-SGD does not introduce a new privacy
accountant, but instead modifies the sum-level private query while remaining
compatible with standard subsampled-Gaussian RDP accounting.

A second line of work improves private optimization itself. Examples include
adaptive clipping \cite{AndrewKMSV21}, private adaptive gradient methods
\cite{AsiDFJT21,LiZLRMS23}, DP-FTRL-style training procedures
\cite{KairouzMSTTX21}, public-data-assisted refinements
\cite{AmidGMRSSTT22}, and adaptive or sharpness-aware private optimizers such
as DP-Adam, DP-IS, SA-DP-SGD, ADP-AdamW, DP-SAT, and DP-Adam-AC
\cite{tang2024dpadambc,WeiBXY22,phan2017adaptive,chilukoti2025differentially,park2023differentially,yang2025dp}.
Other frameworks, such as PATE, improve the privacy--utility trade-off through
changes in supervision and aggregation structure rather than only the optimizer
\cite{PapernotSMR16,PapernotSESMR18}. In contrast to these optimizer-side or
training-protocol refinements, our method introduces memory directly into the
private release mechanism before Gaussian perturbation.

A separate but conceptually related literature studies fractional-order and
history-aware optimization in non-private settings. Fractional calculus
provides principled tools for modeling nonlocal dependence, long memory, and
power-law dynamics
\cite{OldhamS74,MillerR93,SamkoKM93,Podlubny98,KilbasST06,Caputo67}, and
fractional-order optimization methods have been explored as a way to alter
transient behavior, robustness, and stability through structured memory in
iterative updates \cite{WeiKYW18,BaoPZ18,ShengWCW20}. However, this literature
does not address privacy-specific requirements such as clipped sensitivity
control, Poisson subsampling, transcript-dependent releases, Gaussian
mechanism calibration, and formal privacy accounting.

Our work lies at the intersection of these directions but differs in design
emphasis. Unlike accounting-focused methods, we do not propose a new accounting
framework. Unlike optimizer-side modifications, we do not only alter the
parameter update after a private release has been formed. Unlike non-private
fractional optimization methods, we develop structured memory in a setting where
the released statistic itself must remain compatible with sensitivity control
and subsampled-Gaussian privacy analysis. FO-DP-SGD therefore introduces
history dependence directly into the \emph{sum-level private query} while
preserving the classical sum-then-noise-then-divide structure of DP-SGD.

More specifically, FO-DP-SGD combines a recursive private release at the
clipped-sum level, fractional memory over prior \emph{private} sum-level
outputs, an EMA-based private trend for transcript-dependent inconsistency
measurement, and a confidence-aware tempering rule that avoids overly
aggressive reweighting when the trend itself is weak. To our knowledge, this
combination of query-level fractional memory, private-history-dependent
tempering, and compatibility with subsampled-Gaussian DP accounting has not
been jointly developed in prior private optimization work.

\section{Methodology: Fractional-Order Differentially Private Stochastic Gradient Descent}
\label{sec:fodpsgd_methodology}

We present \emph{Fractional-Order Differentially Private Stochastic Gradient
Descent} (FO-DP-SGD), a mechanism-level extension of DP-SGD in which the
private sum-level query combines the current clipped subsampled sum with a
fractional private memory state constructed from previously released private
sum-level outputs. Unlike post-processing-based memory heuristics applied after
a private gradient has already been released, FO-DP-SGD modifies the private
query itself while preserving the classical \emph{sum-then-noise-then-divide}
structure. The method further incorporates an EMA-based private trend, a
normalized inconsistency score, and a confidence-aware scaling factor that
attenuates inconsistency-based tempering when the recent private trend is weak.
The full procedure is summarized in Algorithm~\ref{alg:fodpsgd}, while
conditional sensitivity and privacy accounting are deferred to
Section~\ref{sec:fodpsgd_theory}.

\subsection{Setup}
\label{subsec:fodpsgd_setup}

Let \(D=\{x_i\}_{i=1}^N\) denote the dataset. At iteration \(t\), Poisson
subsampling draws independent indicators \(m_{t,i}\sim\mathrm{Bernoulli}(q)\),
with mask \(m_t=(m_{t,1},\dots,m_{t,N})\in\{0,1\}^N\), sampled set
\(S_t=\{i:m_{t,i}=1\}\), and expected lot size \(L=Nq\). For each
\(i\in S_t\), let \(g_t(x_i)=\nabla_\theta \ell(\theta_t;x_i)\), and clip the
gradient as
\begin{equation}
\bar g_t(x_i)=
\frac{g_t(x_i)}
{\max\!\left(1,\frac{\|g_t(x_i)\|_2}{C}\right)},
\label{eq:fodpsgd_clip}
\end{equation}
so that \(\|\bar g_t(x_i)\|_2\le C\). The clipped subsampled sum is
\begin{equation}
s_t(D;m_t)=\sum_{i\in S_t}\bar g_t(x_i)
=\sum_{i=1}^N m_{t,i}\bar g_t(x_i).
\label{eq:fodpsgd_subsampled_sum}
\end{equation}

We adopt add/remove adjacency throughout.

\begin{definition}[Add/remove adjacency]~\cite{AbadiCGMMTZ16}
\label{def:fodpsgd_adjacency}
Two datasets \(D\) and \(D'\) are adjacent, denoted \(D\sim D'\), if one can
be obtained from the other by adding or removing exactly one example.
\end{definition}

For a fixed sampling mask \(m_t\), define the mask-conditioned sensitivity
\begin{equation}
\Delta_s(m_t)=\sup_{D\sim D'}
\|s_t(D;m_t)-s_t(D';m_t)\|_2.
\label{eq:fodpsgd_delta_s_def}
\end{equation}
Under add/remove adjacency, \(\Delta_s(m_t)\le C\), since for any fixed mask
the two subsampled sums can differ in at most one clipped example contribution.

\subsection{EMA-based confidence-aware inconsistency tempering}
\label{subsec:fodpsgd_ema_confidence}

Let \(h_t=(\tilde s_0,\tilde s_1,\dots,\tilde s_{t-1})\) denote the prior
private transcript, where \(\tilde s_r\) is the released private sum-level
output at iteration \(r\). Let \(K\in\mathbb N\) be the memory window,
\(\alpha\in(0,1]\) the fractional order, \(\lambda\ge 0\) the baseline
tempering parameter, \(\tau\ge 0\) the inconsistency-aware tempering
coefficient, \(\gamma\in(0,1]\) the EMA coefficient, \(\kappa>0\) the minimum
normalization scale, \(\varepsilon>0\) a numerical stability constant, and
\(\zeta>0\) the confidence parameter. Define \(K_t=\min\{K,t+1\}\).

To summarize recent private history causally, we define the EMA-based private
trend by \(\bar s_1^{\mathrm{EMA}}=\tilde s_0\) and, for \(t\ge 2\),
\begin{equation}
\bar s_t^{\mathrm{EMA}}
=
\gamma \tilde s_{t-1}
+
(1-\gamma)\bar s_{t-1}^{\mathrm{EMA}}.
\label{eq:fodpsgd_ema}
\end{equation}
For each lag \(j=1,\dots,K_t-1\), define the inconsistency score
\begin{equation}
\nu_{t,j}
=
\frac{\|\tilde s_{t-j}-\bar s_t^{\mathrm{EMA}}\|_2}
{\max\!\bigl(\|\bar s_t^{\mathrm{EMA}}\|_2,\kappa\bigr)+\varepsilon},
\label{eq:fodpsgd_inconsistency}
\end{equation}
and the confidence factor
\begin{equation}
\chi_t
=
\frac{\|\bar s_t^{\mathrm{EMA}}\|_2}
{\|\bar s_t^{\mathrm{EMA}}\|_2+\zeta},
\qquad 0\le \chi_t<1.
\label{eq:fodpsgd_confidence}
\end{equation}
Thus, inconsistency-aware tempering is weakened when the recent private trend
is itself weak. The raw confidence-aware fractional kernel is
\begin{equation}
a_{t,j}^{\mathrm{CA}}
=
(j+1)^{\alpha-1}
\exp\!\Bigl(-\bigl(\lambda+\chi_t\tau \nu_{t,j}\bigr)j\Bigr),
\qquad j=1,\dots,K_t-1,
\label{eq:fodpsgd_kernel}
\end{equation}
with normalized weights
\begin{equation}
\hat a_{t,j}^{\mathrm{CA}}
=
\frac{a_{t,j}^{\mathrm{CA}}}{\sum_{l=1}^{K_t-1}a_{t,l}^{\mathrm{CA}}},
\qquad
\sum_{j=1}^{K_t-1}\hat a_{t,j}^{\mathrm{CA}}=1.
\label{eq:fodpsgd_norm_weights}
\end{equation}
For \(K_t\ge 2\), the private memory state is
\begin{equation}
u_{t-1}^{\mathrm{CA}}(h_t)
=
\sum_{j=1}^{K_t-1}\hat a_{t,j}^{\mathrm{CA}}\,\tilde s_{t-j},
\label{eq:fodpsgd_memory_state}
\end{equation}
while for \(K_t=1\) we set \(u_{t-1}^{\mathrm{CA}}(h_t)=0\).

\subsection{Recursive private mechanism}
\label{subsec:fodpsgd_mechanism}

Let \(\beta\in(0,1]\) denote the mixing coefficient between the current clipped
subsampled sum and the private memory state. The recursive sum-level private
query is
\begin{equation}
r_t(D;m_t,h_t)
=
\beta\, s_t(D;m_t)
+
(1-\beta)\,u_{t-1}^{\mathrm{CA}}(h_t).
\label{eq:fodpsgd_recursive_query}
\end{equation}
At \(t=0\), one has \(K_t=1\) and hence
\(r_0(D;m_0,h_0)=\beta\, s_0(D;m_0)\). The private sum-level release is
\begin{equation}
\tilde s_t
=
r_t(D;m_t,h_t)+Z_t,
\qquad
Z_t\sim\mathcal N(0,\sigma^2 C^2 I),
\label{eq:fodpsgd_release}
\end{equation}
where \(\sigma>0\) is the noise multiplier. The released noisy gradient is
\(\tilde g_t=\tilde s_t/L\), and the model update is
\begin{equation}
\theta_{t+1}=\theta_t-\eta \tilde g_t,
\label{eq:fodpsgd_update}
\end{equation}
where \(\eta>0\) is the learning rate. The complete proposed
procedure is given in Algorithm~\ref{alg:fodpsgd}.

\begin{algorithm}[H]
\caption{Fractional-Order Differentially Private Stochastic Gradient Descent with EMA-Based Confidence-Aware Inconsistency Tempering}
\label{alg:fodpsgd}
\begin{algorithmic}[1]
\Require Dataset $D=\{x_i\}_{i=1}^N$, loss $\ell(\theta;x)$, clipping norm $C$, subsampling probability $q$, expected lot size $L=Nq$, noise multiplier $\sigma$, learning rate $\eta$, memory window $K$, fractional order $\alpha\in(0,1]$, baseline tempering parameter $\lambda\ge 0$, inconsistency-aware tempering coefficient $\tau\ge 0$, EMA coefficient $\gamma\in(0,1]$, mixing coefficient $\beta\in(0,1]$, minimum normalization scale $\kappa>0$, confidence parameter $\zeta>0$, numerical stability constant $\varepsilon>0$
\State Initialize parameters $\theta_0$, private release buffer $\mathcal B\gets [\,]$, and EMA trend state as undefined
\For{$t=0,1,\dots,T-1$}
    \State Draw $m_{t,i}\sim\mathrm{Bernoulli}(q)$ independently for $i=1,\dots,N$ and define $S_t=\{i:m_{t,i}=1\}$
    \State Compute the clipped subsampled sum $s_t=\sum_{i\in S_t}\bar g_t(x_i)$ using \eqref{eq:fodpsgd_clip}
    \State Set $K_t\gets \min\{K,t+1\}$
    \If{$K_t=1$}
        \State $u_{t-1}^{\mathrm{CA}}\gets 0$, \quad $r_t\gets \beta s_t$
    \Else
        \State Retrieve EMA trend state $\bar s_t^{\mathrm{EMA}}$
        \For{$j=1,\dots,K_t-1$}
            \State Compute $\nu_{t,j}$ and $a_{t,j}^{\mathrm{CA}}$ using \eqref{eq:fodpsgd_inconsistency} and \eqref{eq:fodpsgd_kernel}
        \EndFor
        \State Compute \(\chi_t\) using \eqref{eq:fodpsgd_confidence} and normalize \(a_{t,j}^{\mathrm{CA}}\) to obtain \(\hat a_{t,j}^{\mathrm{CA}}\)
        \State Form \(u_{t-1}^{\mathrm{CA}}=\sum_{j=1}^{K_t-1}\hat a_{t,j}^{\mathrm{CA}}\,\mathcal B[j]\)
        \State Form \(r_t\gets \beta s_t+(1-\beta)u_{t-1}^{\mathrm{CA}}\)
    \EndIf
    \State Draw \(Z_t\sim \mathcal N(0,\sigma^2 C^2 I)\) and set \(\tilde s_t\gets r_t+Z_t\)
    \State Set \(\tilde g_t\gets \tilde s_t/L\) and update \(\theta_{t+1}\gets\theta_t-\eta\tilde g_t\)
    \State Insert \(\tilde s_t\) at the front of \(\mathcal B\); if needed, remove the oldest entry
    \If{$t=0$}
        \State Set \(\bar s_1^{\mathrm{EMA}}\gets \tilde s_0\)
    \Else
        \State Update \(\bar s_{t+1}^{\mathrm{EMA}}\gets \gamma \tilde s_t+(1-\gamma)\bar s_t^{\mathrm{EMA}}\)
    \EndIf
\EndFor
\State \Return \(\theta_T\)
\end{algorithmic}
\end{algorithm}

\paragraph{Deferred analysis.}
The conditional sensitivity argument, compatible R\'enyi privacy accounting
under Poisson subsampling, and interpretation of parameter regimes are deferred
to Section~\ref{sec:fodpsgd_theory}.


\section{Experiments}
\label{sec:experiment}

We empirically evaluate FO-DP-SGD on multiple benchmark datasets and privacy
configurations to assess four questions: (i) whether FO-DP-SGD improves utility
relative to standard DP-SGD and related private optimization baselines,
(ii) whether its gains arise from the proposed mechanism-level recursive query
rather than post-processing alone, (iii) how the key parameters
\(\beta\), \(\alpha\), and \(K\) affect the privacy--utility trade-off, and
(iv) what computational overhead is introduced by the proposed method.

\subsection{Experimental Setup and Reproducibility}
\label{subsec:experimental-setup-reproducibility}

We evaluate all methods on SVHN, CIFAR-10, and CIFAR-100 under the same differentially private training protocol. 
For each dataset, we use a fixed train/test subset and apply the same preprocessing pipeline across all methods. 
All input images are normalized and flattened before being passed to the classifier. 
Unless otherwise stated, experiments use a fixed training subset of 5,000 examples and a fixed test subset of 2,000 examples. All methods are evaluated using the same neural network architecture to ensure a fair comparison. 
The model consists of a fully connected classifier with input dimension matching the flattened image size, two hidden layers with 64 and 32 units, respectively, and a final classification layer. 
The hidden layers use hyperbolic tangent activations. 
The same initialization protocol is used across all methods and random seeds. We use Gaussian-noise-based differentially private training following the standard DP-SGD framework~\cite{AbadiCGMMTZ16}. 
Unless otherwise stated, we use clipping norm $C=1.0$, noise multiplier $\sigma=1.1$, sampling rate $q=0.04$, and privacy parameter $\delta=10^{-5}$. 
The accumulated privacy budget $\varepsilon$ is computed using a Rényi differential privacy accountant. 
For FO-DP-SGD, the effective privacy accounting depends on the fractional mixing parameter $\beta$, and we report the corresponding accumulated $\varepsilon$ for each configuration. The main FO-DP-SGD configuration uses $\beta=0.90$, fractional order $\alpha=0.80$, and memory depth $K=8$. 
The learning rate is set to $0.8$ for FO-DP-SGD. 
For the ablation studies, we vary one parameter at a time while keeping the remaining parameters fixed. 
Specifically, we evaluate $\beta \in \{1.00,0.95,0.90,0.80,0.65\}$, $\alpha \in \{0.20,0.50,0.80\}$, and $K \in \{1,2,8,32\}$. We compare the proposed method against several differentially private optimization baselines, including DP-SGD~\cite{AbadiCGMMTZ16}, DP-Adam~\cite{tang2024dpadambc}, DP-IS~\cite{WeiBXY22}, SA-DP-SGD~\cite{phan2017adaptive}, ADP-AdamW~\cite{chilukoti2025differentially}, DP-SAT~\cite{park2023differentially}, and DP-Adam-AC~\cite{yang2025dp}. 
In our implementation, DP-SGD corresponds to the FO-DP-SGD variant with $\beta=1.00$, which removes the fractional privacy--memory modulation and serves as the non-fractional baseline. 
DP-Adam and ADP-AdamW represent adaptive moment-based private optimization methods, DP-IS incorporates importance-aware private sampling, SA-DP-SGD and DP-SAT represent adaptive or sharpness-aware private training strategies, and DP-Adam-AC uses adaptive clipping within an Adam-style private optimizer. 
All baselines are trained under the same dataset split, model architecture, privacy parameters, and evaluation protocol whenever applicable. 
Baseline-specific hyperparameters are selected according to their best validation behavior under the considered experimental budget. All main comparison experiments are repeated over five independent random seeds. 
We report mean $\pm$ standard deviation across seeds. 
For final test accuracy, we additionally report 95\% confidence intervals computed using Student's $t$ interval $\bar{x} \pm t_{0.975,n-1}\frac{s}{\sqrt{n}},$ where $\bar{x}$ is the sample mean, $s$ is the sample standard deviation, and $n$ is the number of independent runs. 
Training curves report the representative trajectory unless otherwise specified, while final comparison tables summarize the multi-seed statistics. To support reproducibility, we keep the dataset splits, model architecture, privacy parameters, optimizer hyperparameters, and random seeds fixed across methods. 
The implementation records per-run final accuracy, best accuracy, final loss, final privacy budget, and runtime. 
The multi-seed summary tables are generated directly from the raw seed-level results, which include the algorithm name, seed, final accuracy, best accuracy, final loss, final $\varepsilon$, and runtime. 
This ensures that the reported mean, standard deviation, and confidence intervals can be regenerated from the saved experimental logs.
\subsection{Cross-Dataset Accuracy Comparison}
\label{subsec:cross-dataset-accuracy}

Figure~\ref{fig:cross-dataset-accuracy} compares the test accuracy trajectories of the proposed FO-DP-SGD method against several differentially private optimizer baselines across three benchmark datasets: SVHN, CIFAR-10, and CIFAR-100. 
The evaluated baselines include DP-SGD/FO-DP-SGD with $\beta=1.00$, DP-Adam, DP-AdamW, DP-IS, DP-SAM, DP-SAT, and DP-Adam-AC. 
The proposed configuration uses $\beta=0.90$, $\alpha=0.80$, and $K=8$.

\begin{figure}[H]
\centering
\includegraphics[height=3.5cm,width=4.5cm]{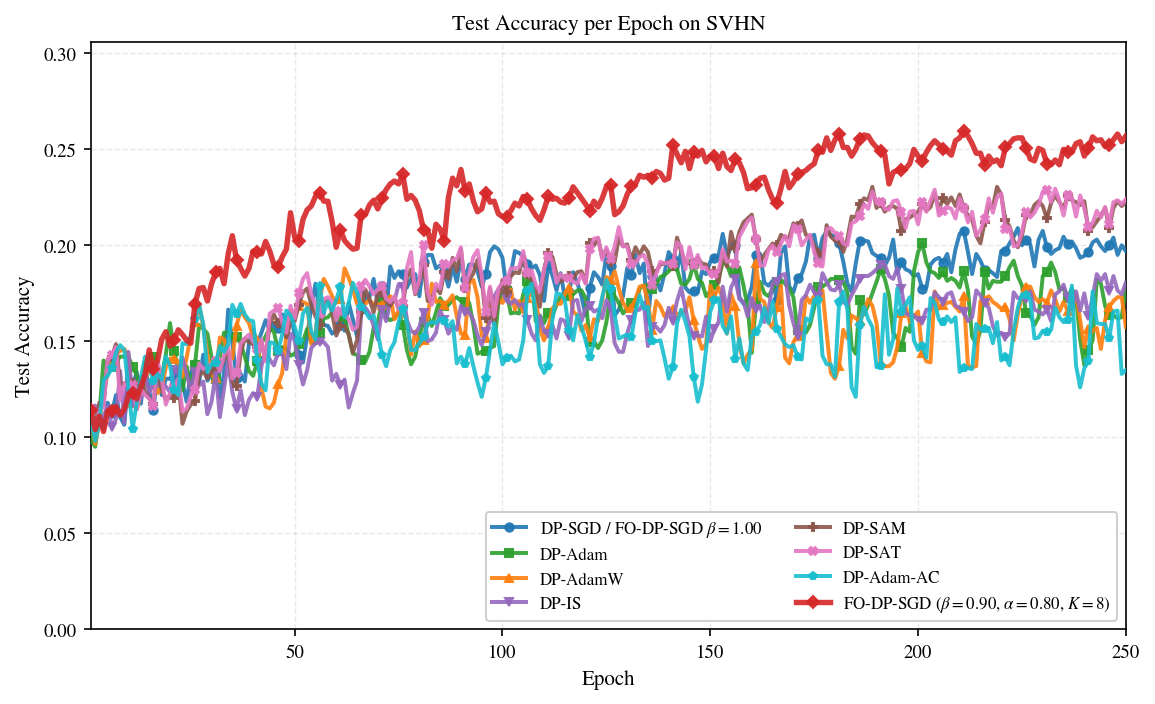}\includegraphics[height=3.5cm,width=4.5cm]{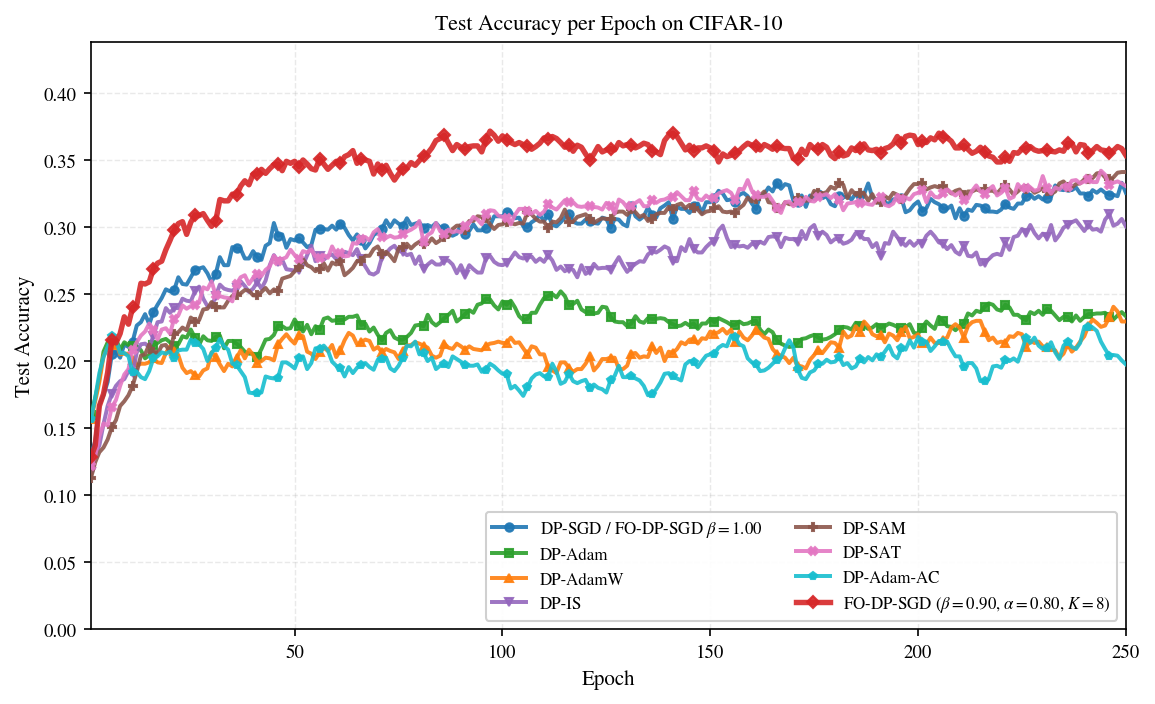}\includegraphics[height=3.5cm,width=4.5cm]{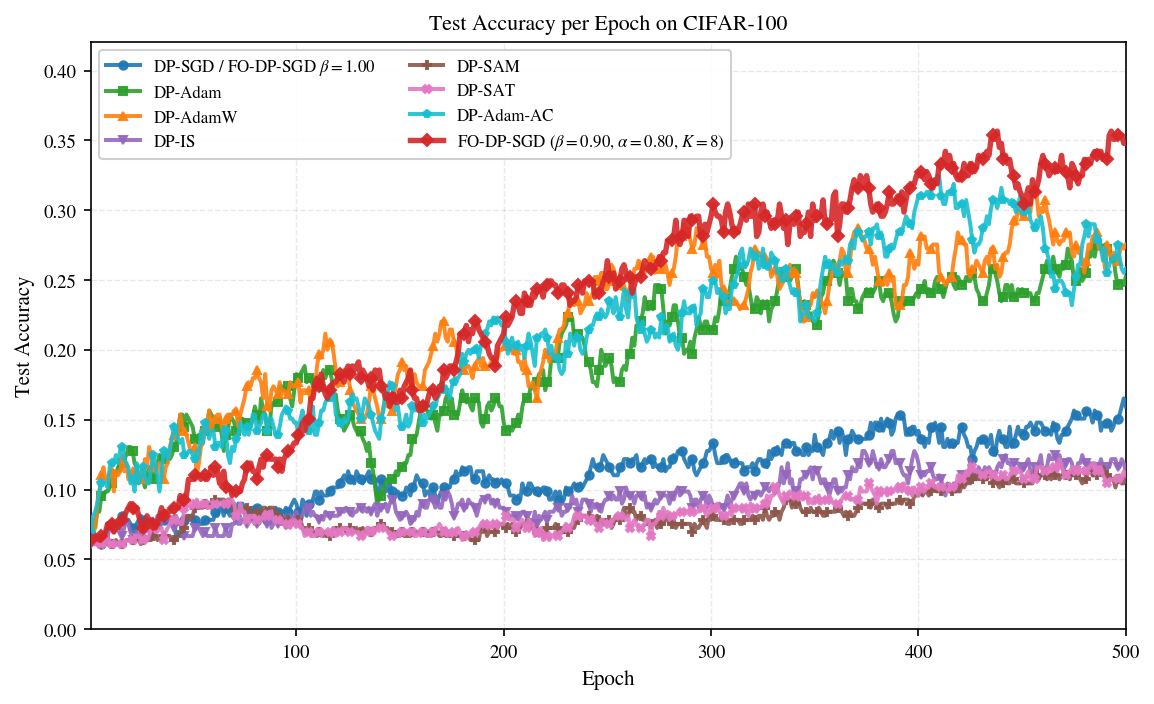}
    \caption{
    Cross-dataset comparison of test accuracy over training epochs on SVHN, CIFAR-10, and CIFAR-100. 
    The proposed FO-DP-SGD method with $\beta=0.90$, $\alpha=0.80$, and $K=8$ consistently achieves the strongest or among the strongest accuracy trajectories across all datasets, demonstrating improved robustness under differentially private training.
    }
\label{fig:cross-dataset-accuracy}
\end{figure}
As shown in Figure~\ref{fig:cross-dataset-accuracy}, the proposed FO-DP-SGD method consistently provides strong accuracy performance across all evaluated datasets. 
On SVHN, the proposed method achieves the highest accuracy trajectory among the compared methods for most of the training process. 
Although the absolute accuracy is lower than in CIFAR-10 due to the specific experimental setting and private training noise, the proposed method shows a clear advantage over the competing DP optimizers, particularly during the middle and late training stages.

On CIFAR-10, the proposed FO-DP-SGD configuration rapidly improves during the early epochs and reaches the highest overall accuracy. 
Its trajectory remains consistently above the DP-Adam, DP-AdamW, DP-SAM, DP-SAT, DP-Adam-AC, and DP-SGD/FO-DP-SGD $\beta=1.00$ baselines. 

The CIFAR-100 results further highlight the advantage of the proposed method in a more challenging classification setting. 
Because CIFAR-100 contains a larger number of classes, the optimization problem is more difficult and the accuracy curves exhibit stronger fluctuations. 
Nevertheless, the proposed FO-DP-SGD method achieves the best final trajectory and maintains a clear advantage over the standard DP-SGD/FO-DP-SGD $\beta=1.00$ baseline as well as the adaptive DP optimizer variants. 

\subsection{Cross-Dataset Privacy--Utility Tradeoff over the Fractional Mixing Parameter}
\label{subsec:cross-dataset-privacy-utility-beta}

Figure~\ref{fig:cross-dataset-privacy-utility-beta} compares the privacy--utility tradeoff of FO-DP-SGD under different values of the fractional mixing parameter $\beta$ across CIFAR-10, CIFAR-100, and SVHN. 
For each dataset, the horizontal axis reports the accumulated privacy budget $\varepsilon$, while the vertical axis reports test accuracy. 
We evaluate $\beta \in \{1.00,0.95,0.90,0.80,0.65\}$ while keeping the remaining optimization and privacy parameters fixed within each dataset.

\begin{figure}[H]
\centering
\includegraphics[height=3.5cm,width=4.5cm]{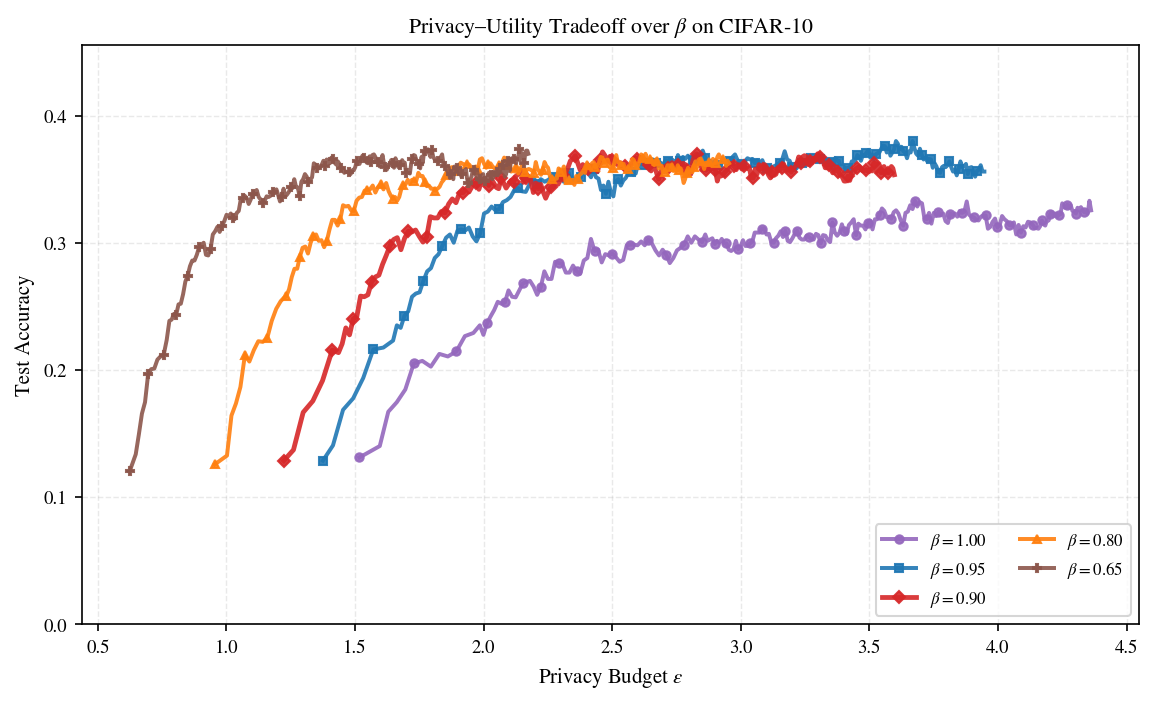}\includegraphics[height=3.5cm,width=4.5cm]{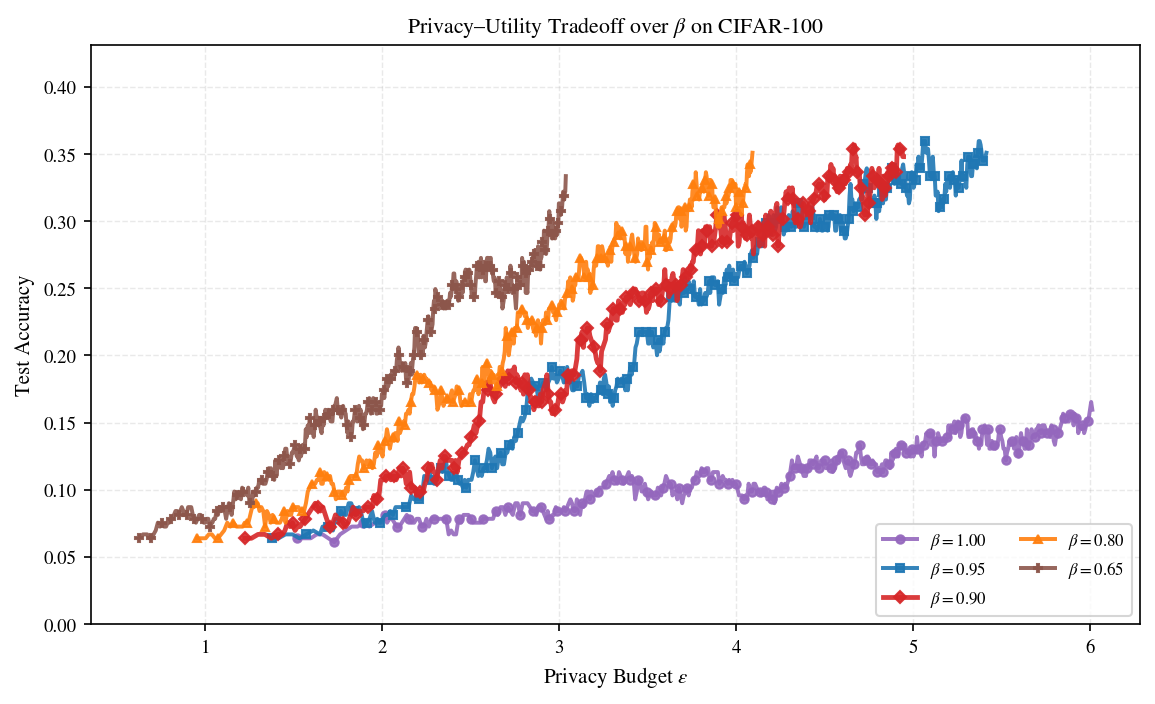}\includegraphics[height=3.5cm,width=4.5cm]{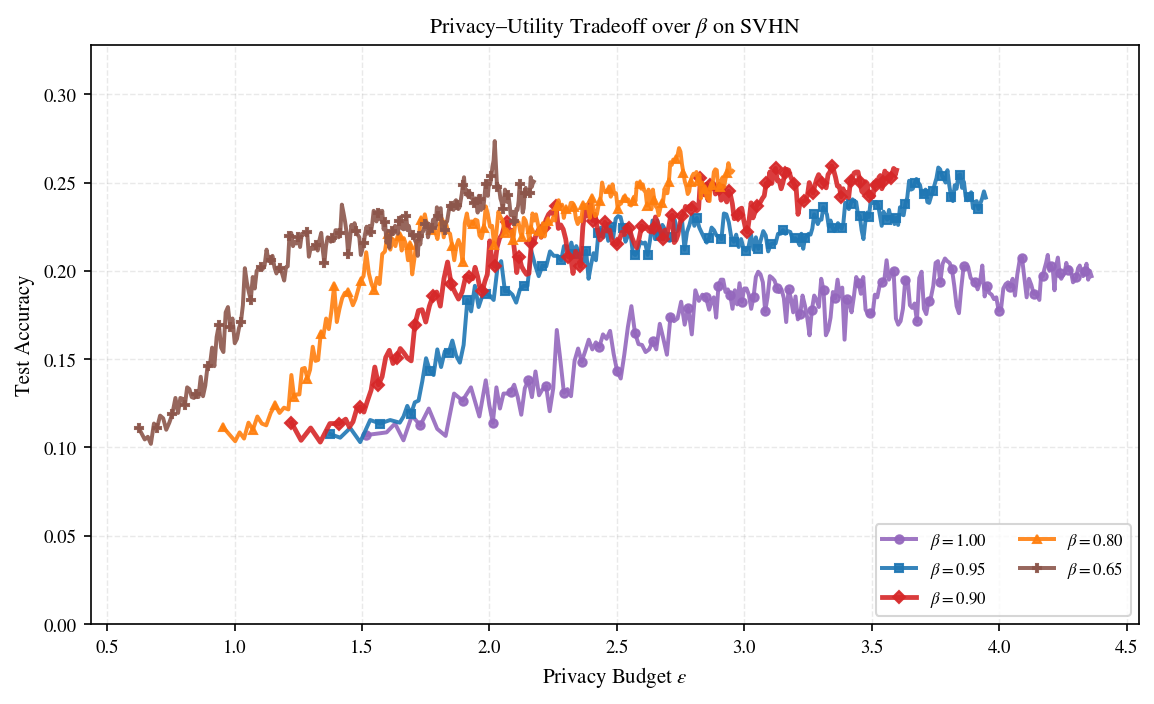}
    \caption{
    Cross-dataset privacy--utility tradeoff over the fractional mixing parameter $\beta$ on CIFAR-10, CIFAR-100, and SVHN. 
    Each curve shows test accuracy as a function of the accumulated privacy budget $\varepsilon$. 
    }
 \label{fig:cross-dataset-privacy-utility-beta}
\end{figure}

As shown in Figure~\ref{fig:cross-dataset-privacy-utility-beta}, the fractional mixing parameter $\beta$ strongly influences the tradeoff between model utility and privacy expenditure across all three datasets. 
The baseline configuration $\beta=1.00$ generally produces weaker tradeoff curves, achieving lower accuracy for a given privacy budget compared with the fractional-memory variants. 
This behavior is especially visible on CIFAR-100 and SVHN, where the $\beta=1.00$ curve remains substantially below the curves obtained with $\beta<1$ throughout most of training. 

On CIFAR-10, the fractional configurations achieve high accuracy at substantially smaller privacy budgets than the $\beta=1.00$ baseline. 
In particular, $\beta=0.65$ and $\beta=0.80$ rapidly improve accuracy at low $\varepsilon$, showing that stronger memory contribution can accelerate utility gain under limited privacy expenditure. 
Meanwhile, $\beta=0.90$ and $\beta=0.95$ also reach competitive high-accuracy regimes, providing a smoother transition toward the best final utility. 

On CIFAR-100, the privacy--utility tradeoff is more difficult because the classification task is substantially more complex. 
The $\beta=1.00$ curve remains far below the other variants, while $\beta=0.80$, $\beta=0.90$, and $\beta=0.95$ progressively achieve higher accuracy as the privacy budget increases. 
The lower value $\beta=0.65$ reaches competitive accuracy at a smaller privacy budget, suggesting that stronger memory weighting can be especially useful when privacy cost must be tightly controlled. 
However, intermediate settings such as $\beta=0.90$ remain attractive because they provide strong final utility while still reducing $\varepsilon$ relative to the $\beta=1.00$ baseline.

On SVHN, the same general pattern appears: all fractional-memory variants outperform the $\beta=1.00$ baseline in the privacy--utility plane. 
The $\beta=0.65$ and $\beta=0.80$ settings achieve strong early utility at relatively low privacy budgets, while $\beta=0.90$ and $\beta=0.95$ provide competitive accuracy in the later stages. 
This consistency across SVHN, CIFAR-10, and CIFAR-100 suggests that the role of $\beta$ is not dataset-specific; rather, it acts as a general mechanism for controlling the balance between immediate private gradient information and stabilizing fractional-memory effects.



\section{Conclusion and Future Work}
\label{sec:conclusion}

We introduced FO-DP-SGD, a fractional-order recursive extension of DP-SGD that incorporates memory directly into the private release mechanism. Unlike optimizer-level modifications applied only after privatization, FO-DP-SGD modifies the sum-level private query itself and integrates historical information through a fractional, tempered, and confidence-aware memory construction.

Our analysis showed that the proposed mechanism remains compatible with standard differential privacy accounting under Poisson subsampling, while the recursive query admits a history-conditioned sensitivity bound scaled by \(\beta\). This preserves the standard accounting pipeline while enabling a mechanism-level trade-off between current-step sensitivity control and memory-enhanced optimization dynamics under privacy noise.

Empirically, FO-DP-SGD improved convergence behavior, predictive utility, and privacy--utility performance relative to standard DP-SGD and the evaluated private baselines. In particular, across the considered experimental settings, it achieved stronger accuracy under fixed privacy budgets and required fewer training epochs to reach target utility levels.

\paragraph{Limitations and future work.}
FO-DP-SGD mitigates unreliable history through fractional weighting,
inconsistency-aware tempering, and confidence-aware scaling, but it does not
explicitly denoise previous private releases. Since each historical release
\(\tilde{s}_{t-j}\) contains both the recursive query signal and injected
Gaussian noise, the memory state may inherit useful directional information as
well as accumulated release noise. Future work will therefore investigate
explicit transcript-based denoising or filtering mechanisms for private memory,
such as reliability-weighted shrinkage, uncertainty-aware filtering, or
Kalman-style smoothing applied only to previously released private quantities.
Because such operations can be defined as functions of the private transcript
and public hyperparameters, they may be treated as post-processing and need not
consume additional privacy budget, while potentially improving the stability of
long-memory private optimization.
\section{Broader Impact}
\label{sec:broader-impact}

This work contributes to privacy-preserving machine learning by improving the utility of differentially private optimization. A positive societal impact is that more accurate private training methods may support learning from sensitive data in domains such as healthcare, finance, education, and mobile or behavioral analytics without requiring unrestricted exposure of individual-level records.

At the same time, formal privacy guarantees must be interpreted carefully. FO-DP-SGD provides privacy protection only under the stated clipping, sampling, noise, adjacency, and accounting assumptions. Incorrect privacy accounting, inappropriate hyperparameter choices, excessive privacy budgets, or deployment under distribution shift may weaken the intended protection. In addition, improved private optimization does not by itself address fairness, robustness, model misuse, or all forms of information leakage. Therefore, practical deployments should combine formal privacy accounting with careful validation, auditing, and application-specific risk assessment.


\bibliographystyle{plainnat}
\bibliography{references}

\newpage

\clearpage
\appendix

\section*{Appendix Overview}
\label{app:overview}
\addcontentsline{toc}{section}{Appendix Overview}

This appendix provides additional experimental evidence and theoretical analysis
supporting the main paper. To improve readability and navigation, the
supplementary material is organized into two main parts. 
Appendix~\ref{app:additional-results} contains additional empirical results,
ablation studies, robustness analyses, runtime comparisons, and statistical
summaries. Appendix~\ref{app:theoretical-analysis} provides the theoretical
analysis and interpretation of FO-DP-SGD, including privacy validity,
mechanism-level novelty, recursive release decomposition, parameter
interpretations, and limiting regimes.

\paragraph{Appendix~\ref{app:additional-results}: Additional Experimental Results.}
This part complements the main experimental section with supplementary results.
It includes ablation studies on fractional memory depth, mechanism-level memory
insertion, query-level memory design, clipping-norm sensitivity,
privacy--utility behavior, runtime overhead, and final-accuracy statistics
reported with mean $\pm$ standard deviation and 95\% confidence intervals.

\paragraph{Appendix~\ref{app:theoretical-analysis}: Theoretical Analysis and Interpretation.}
This part analyzes the FO-DP-SGD mechanism from a theoretical perspective. It
establishes the history-conditioned query structure, sensitivity of the
recursive query, privacy accounting under the Poisson-subsampled Gaussian
mechanism, adaptive R\'enyi DP composition, signal--memory--noise decomposition,
and the roles of the main parameters
$\alpha,\lambda,\tau,\gamma,\kappa,\zeta,\beta$, and $K$.

\paragraph{Navigation Guide.}
For convenience, the appendix is organized as follows:
\begin{itemize}

    \item Appendix~\ref{app:additional-results}: supplementary experimental results and ablation studies.
    \item Appendix~\ref{subsec:joint-beta-utility-privacy-cifar100}: joint effect of the fractional mixing parameter on utility and privacy.
    \item Appendix~\ref{subsec:k-ablation-cifar10}: effect of fractional memory depth.
    \item Appendix~\ref{subsec:results_mechanism_privacy_utility}: mechanism-level privacy--utility advantage.
    \item Appendix~\ref{subsec:results_mechanism_training_utility}: mechanism-level memory improves training utility.
    \item Appendix~\ref{subsec:runtime_comparison_cifar10}: absolute runtime comparison on CIFAR-10.
    \item Appendix~\ref{subsec:relative_runtime_overhead_cifar10}: relative runtime overhead on CIFAR-10.
    \item Appendix~\ref{subsec:fixed_sigma_varying_clipping}: effect of clipping norm under fixed noise.
    \item Appendix~\ref{subsec:final-accuracy-analysis}: final accuracy analysis with mean $\pm$ standard deviation and 95\% confidence intervals.
    \item Appendix~\ref{subsec:ablation_query_memory_design}: ablation on query-level memory design.
    \item Appendix~\ref{app:theoretical-analysis}: theoretical analysis and interpretation of FO-DP-SGD.
    \item Appendix~\ref{subsec:fodpsgd_theory_privacy_accounting}: privacy validity, mechanism-level novelty, and privacy accounting.
    \item Appendix~\ref{subsec:fodpsgd_theory_decomposition}: signal--memory--noise decomposition.
    \item Appendix~\ref{subsec:fodpsgd_theory_interpretation}: interpretation of EMA-based confidence-aware tempering.
    \item Appendix~\ref{subsec:fodpsgd_theory_special_cases}: special cases and limiting regimes.
\end{itemize}

\clearpage

\section{Additional Experimental Results}
\label{app:additional-results}

This section provides supplementary empirical results that support the main
experimental findings. The additional results examine how FO-DP-SGD behaves
under different memory depths, memory insertion mechanisms, query-level memory
designs, clipping norms, runtime constraints, and multi-seed statistical
evaluations. These experiments clarify which components of the proposed method
are responsible for the observed privacy--utility improvements.

\subsection{Joint Effect of the Fractional Mixing Parameter on Utility and Privacy}
\label{subsec:joint-beta-utility-privacy-cifar100}

Figure~\ref{fig:joint-beta-utility-privacy-cifar100} jointly analyzes the effect of the fractional mixing parameter $\beta$ on both test accuracy and privacy expenditure on CIFAR-100. 
The left panel reports test accuracy as a function of training epoch, while the right panel reports the accumulated privacy budget $\varepsilon$ over the same training horizon. 
We evaluate $\beta \in \{1.00,0.95,0.90,0.80,0.65\}$ while keeping all other optimization and privacy parameters fixed.

\begin{figure}[H]
\centering
\includegraphics[height=4.5cm,width=5cm]{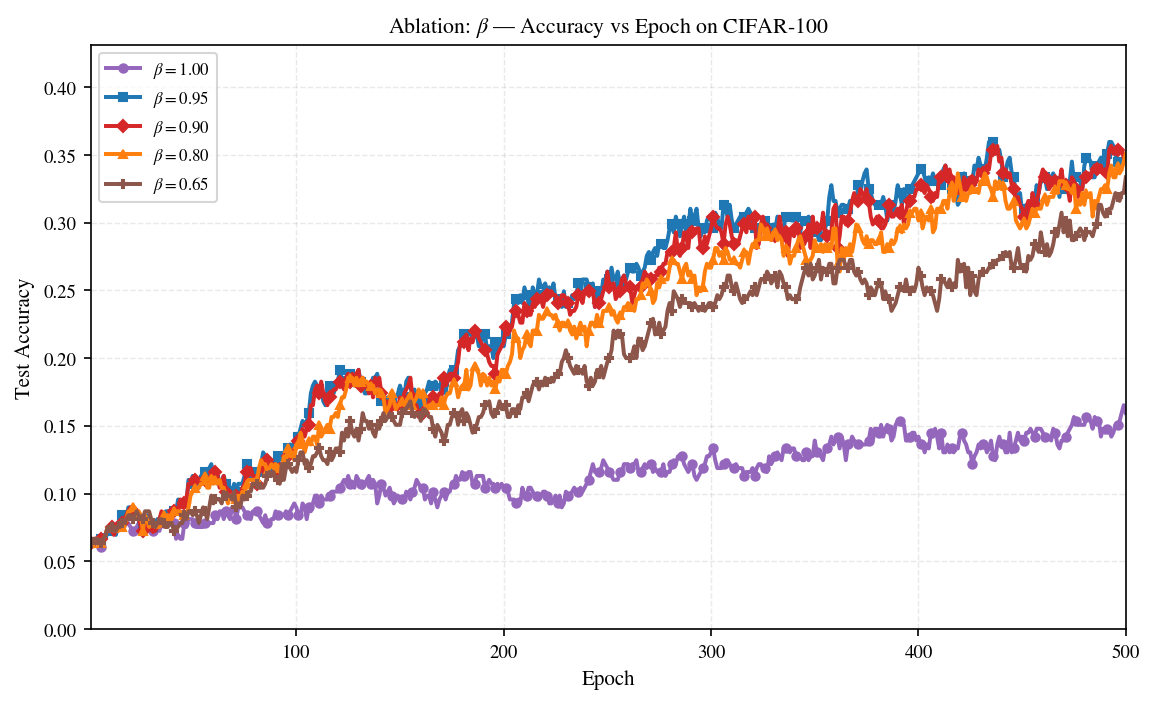}\includegraphics[height=4.5cm,width=5cm]{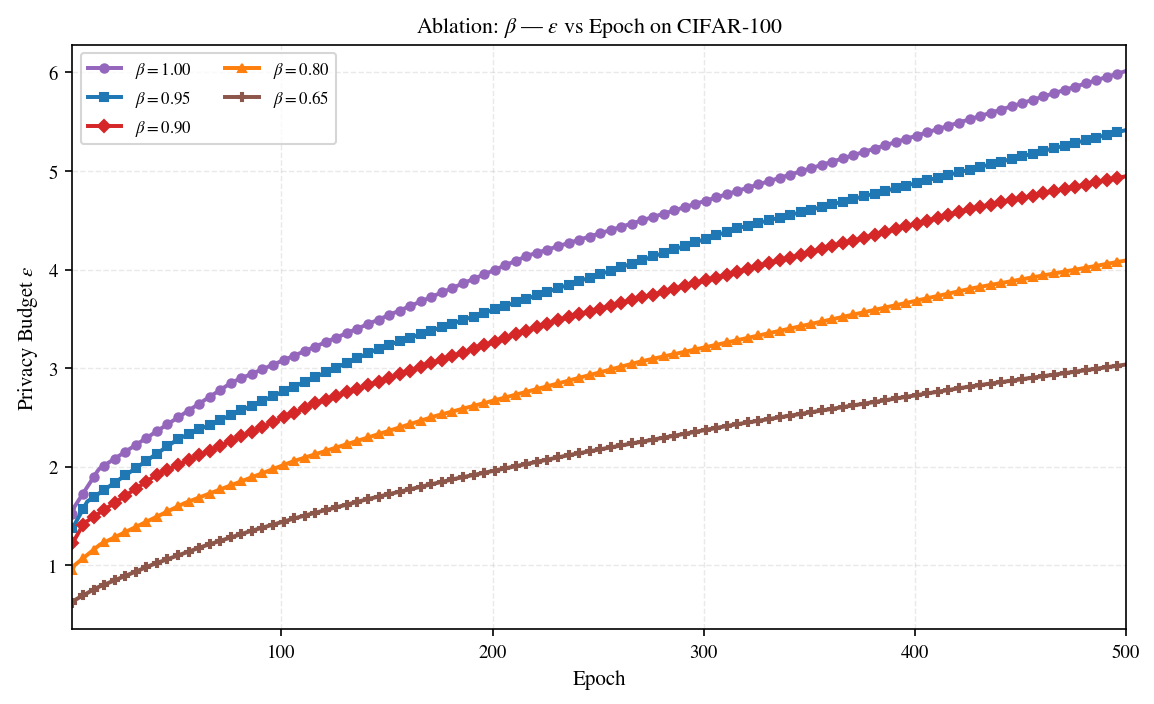}
    \caption{    Joint ablation of the fractional mixing parameter $\beta$ on CIFAR-100. 
    Left: test accuracy versus training epoch. 
    Right: accumulated privacy budget $\varepsilon$ versus training epoch. 
    }
 \label{fig:joint-beta-utility-privacy-cifar100}
\end{figure}
As shown in the left panel of Figure~\ref{fig:joint-beta-utility-privacy-cifar100}, the fractional-memory variants with $\beta<1$ consistently outperform the $\beta=1.00$ baseline in terms of test accuracy. 
The $\beta=1.00$ setting, which corresponds to the weakest-memory or non-fractional behavior, remains substantially below the other curves throughout training and reaches a much lower final accuracy. 
This indicates that relying primarily on the current noisy clipped gradient is insufficient for effective optimization under differential privacy, especially on the more challenging CIFAR-100 classification task.

Among the fractional-memory settings, $\beta=0.95$, $\beta=0.90$, and $\beta=0.80$ achieve the strongest accuracy trajectories. 
These configurations improve more rapidly than the baseline and maintain higher test accuracy during the middle and late stages of training. 
The $\beta=0.65$ setting also improves over $\beta=1.00$, but it generally lags behind the intermediate $\beta$ values, suggesting that excessive down-weighting of the current gradient branch may reduce responsiveness to the present optimization direction.

The right panel of Figure~\ref{fig:joint-beta-utility-privacy-cifar100} shows that the privacy budget $\varepsilon$ increases monotonically with the number of training epochs for every value of $\beta$, as expected under repeated differentially private updates. 
However, the growth rate of $\varepsilon$ is strongly controlled by $\beta$. 
The largest privacy cost is obtained when $\beta=1.00$, while progressively smaller values of $\beta$ yield lower accumulated privacy budgets. 
In particular, $\beta=0.65$ produces the smallest final $\varepsilon$, followed by $\beta=0.80$, $\beta=0.90$, and $\beta=0.95$.

\subsection{Effect of Fractional Memory Depth}
\label{subsec:k-ablation-cifar10}

Figure~\ref{fig:k-ablation-cifar10} evaluates the effect of the fractional
memory depth $K$ on the test accuracy of FO-DP-SGD on CIFAR-10. The parameter
$K$ controls how many previous fractional-memory directions are retained and
incorporated into the current private update. We compare
$K \in \{1,2,8,32\}$, where $K=1$ corresponds to the no-memory setting.

\begin{figure}[H]
    \centering
    \includegraphics[width=0.82\linewidth]{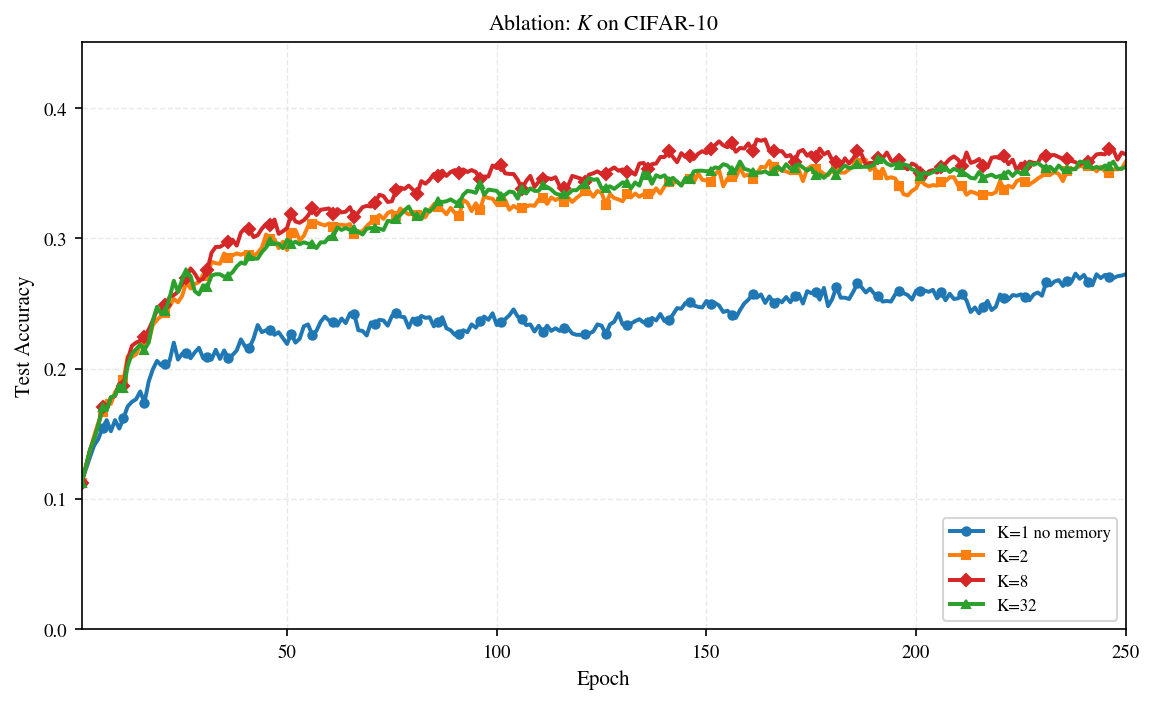}
    \caption{
    Ablation study of the fractional memory depth $K$ on CIFAR-10. 
    The curve with $K=1$ corresponds to the no-memory variant. 
    Increasing $K$ improves accuracy by incorporating longer-range optimization
    history, while an intermediate memory depth, particularly $K=8$, provides
    the strongest overall performance.
    }
    \label{fig:k-ablation-cifar10}
\end{figure}

As shown in Figure~\ref{fig:k-ablation-cifar10}, the no-memory configuration
$K=1$ performs substantially worse than the memory-enhanced variants. Although
$K=1$ improves during the early training stage, its accuracy quickly saturates
at a noticeably lower level compared with $K=2$, $K=8$, and $K=32$. This
confirms that relying only on the current noisy clipped gradient is insufficient
for achieving strong utility under differentially private training.

Introducing fractional memory improves the optimization trajectory. Both $K=2$
and $K=32$ significantly outperform the no-memory baseline across nearly all
epochs, indicating that historical gradient information helps stabilize learning
in the presence of DP noise. The memory mechanism acts as a smoothing component
that aggregates useful directional information from previous updates, thereby
reducing the instability caused by stochastic sampling and Gaussian
perturbations.

Among the evaluated memory depths, $K=8$ achieves the strongest overall
accuracy profile. It rises quickly during the early phase of training and
remains among the top-performing configurations throughout the later epochs.
Compared with $K=2$, the $K=8$ setting benefits from a richer memory window,
which provides more stable long-range optimization information. Compared with
$K=32$, however, $K=8$ avoids excessive reliance on older historical directions,
which may become less aligned with the current optimization landscape.

These results suggest that the memory depth $K$ introduces an important
tradeoff. A very small memory window may not sufficiently reduce the variance of
private updates, while an overly long memory window may include stale gradient
information. An intermediate value such as $K=8$ provides a favorable balance
between stability and adaptivity, improving accuracy without over-smoothing the
optimization dynamics. This supports the use of $K=8$ as the default memory
depth in our main FO-DP-SGD configuration.

\subsection{Mechanism-Level Privacy--Utility Advantage}
\label{subsec:results_mechanism_privacy_utility}

To evaluate whether the benefit of FO-DP-SGD comes from inserting memory
directly into the private query, rather than merely smoothing already privatized
gradients, we compare three methods under the same private training setting:
standard DP-SGD, the post-processing fractional-memory baseline
(Post-FM-DP-SGD), and the proposed FO-DP-SGD mechanism.

Standard DP-SGD first forms the clipped subsampled sum $s_t(D;m_t)$, adds
Gaussian noise, and then normalizes by the expected lot size $L$:
\begin{equation}
\tilde s_t^{\mathrm{std}}
=
s_t(D;m_t)+Z_t,
\qquad
Z_t\sim \mathcal N(0,\sigma^2 C^2 I),
\qquad
\tilde g_t^{\mathrm{std}}
=
\frac{\tilde s_t^{\mathrm{std}}}{L}.
\label{eq:std_dpsgd_release_results}
\end{equation}

The post-processing baseline, denoted Post-FM-DP-SGD, applies fractional memory
only after the standard private release in
\eqref{eq:std_dpsgd_release_results} has already been formed. Specifically, it
constructs the post-processing memory state from previously released private
gradients,
\begin{equation}
u_{t-1}^{\mathrm{post}}
=
\sum_{j=1}^{K_t-1}
\hat a_{t,j}^{\mathrm{post}}
\tilde g_{t-j}^{\mathrm{std}},
\label{eq:post_fm_memory_results}
\end{equation}
where the normalized fractional-memory weights are computed using the same
confidence-aware tempered rule as in FO-DP-SGD:
\begin{equation}
\hat a_{t,j}^{\mathrm{post}}
=
\frac{
(j+1)^{\alpha-1}
\exp\!\left(-(\lambda+\chi_t^{\mathrm{post}}\tau\nu_{t,j}^{\mathrm{post}})j\right)
}{
\sum_{\ell=1}^{K_t-1}
(\ell+1)^{\alpha-1}
\exp\!\left(-(\lambda+\chi_t^{\mathrm{post}}\tau\nu_{t,\ell}^{\mathrm{post}})\ell\right)
}.
\label{eq:post_fm_weights_results}
\end{equation}
The final post-processing update direction is then
\begin{equation}
v_t^{\mathrm{post}}
=
\beta \tilde g_t^{\mathrm{std}}
+
(1-\beta)u_{t-1}^{\mathrm{post}},
\label{eq:post_fm_update_direction_results}
\end{equation}
and the parameter update is
\begin{equation}
\theta_{t+1}^{\mathrm{post}}
=
\theta_t
-
\eta_{\mathrm{post}} v_t^{\mathrm{post}}.
\label{eq:post_fm_parameter_update_results}
\end{equation}
Equivalently, substituting
$\tilde g_t^{\mathrm{std}}=(s_t(D;m_t)+Z_t)/L$, Post-FM-DP-SGD can be written as
\begin{equation}
\theta_{t+1}^{\mathrm{post}}
=
\theta_t
-
\eta_{\mathrm{post}}
\left[
\beta
\frac{s_t(D;m_t)+Z_t}{L}
+
(1-\beta)
\sum_{j=1}^{K_t-1}
\hat a_{t,j}^{\mathrm{post}}
\tilde g_{t-j}^{\mathrm{std}}
\right].
\label{eq:post_fm_expanded_update_results}
\end{equation}
Thus, Post-FM-DP-SGD does not modify the private query itself; it only filters
the already privatized gradients.

By contrast, FO-DP-SGD inserts memory before Gaussian perturbation by forming
the memory-enhanced sum-level query
\begin{equation}
r_t^{\mathrm{FO}}
=
\beta s_t(D;m_t)
+
(1-\beta)u_{t-1}^{\mathrm{FO}},
\qquad
u_{t-1}^{\mathrm{FO}}
=
\sum_{j=1}^{K_t-1}
\hat a_{t,j}^{\mathrm{FO}}\tilde s_{t-j},
\label{eq:fo_query_results}
\end{equation}
and then releases
\begin{equation}
\tilde s_t^{\mathrm{FO}}
=
r_t^{\mathrm{FO}}+Z_t,
\qquad
\theta_{t+1}^{\mathrm{FO}}
=
\theta_t
-
\eta_{\mathrm{FO}}
\frac{\tilde s_t^{\mathrm{FO}}}{L}.
\label{eq:fo_release_update_results}
\end{equation}
Therefore, the two methods use the same class of fractional-memory rules, but
they insert memory at different points: FO-DP-SGD modifies the sum-level private
query before noise, while Post-FM-DP-SGD applies memory only after the standard
private release.

\begin{figure}[H]
    \centering
    \includegraphics[width=0.82\linewidth]{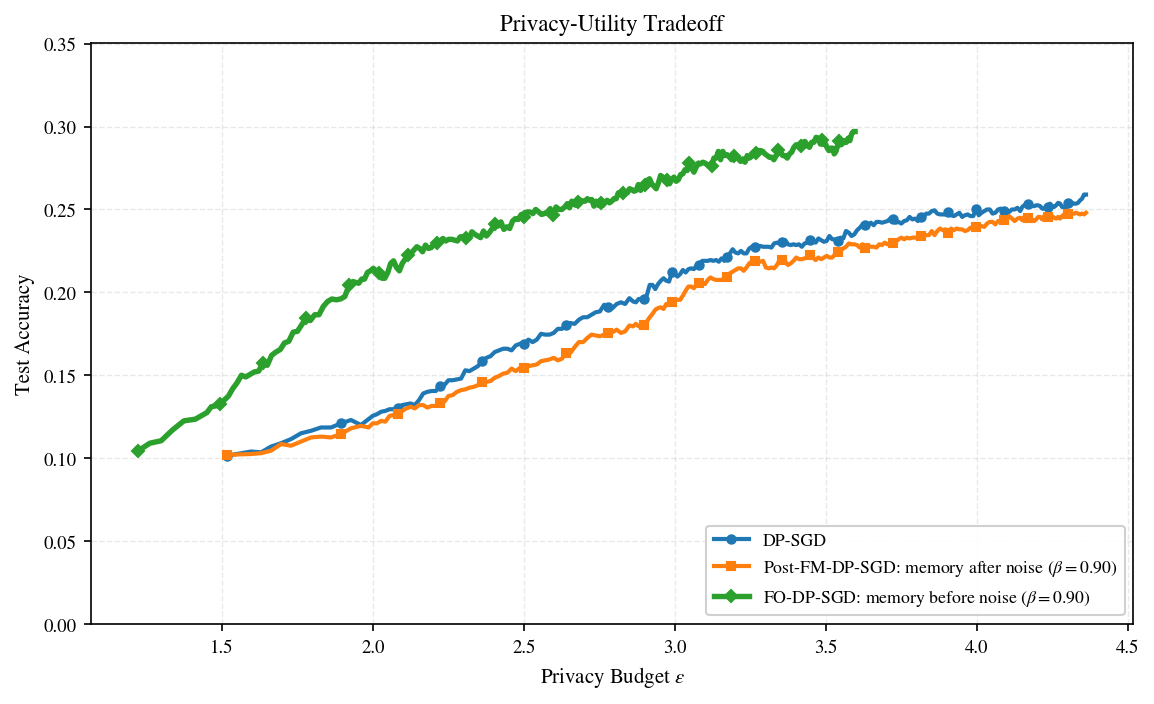}
    \caption{
    \textbf{Privacy--utility tradeoff for mechanism-level memory insertion.}
    Test accuracy is plotted as a function of the privacy budget $\varepsilon$
    on CIFAR-10. DP-SGD serves as the no-memory private baseline. Post-FM-DP-SGD
    applies fractional memory after the standard DP-SGD private release, whereas
    FO-DP-SGD inserts the memory term directly into the sum-level private query
    before Gaussian perturbation. FO-DP-SGD achieves higher accuracy in the
    low-to-moderate privacy-budget regime, supporting the mechanism-level
    advantage of memory-before-noise over post-processing memory-after-noise.
    }
    \label{fig:privacy_utility_mechanism_level}
\end{figure}

As shown in Figure~\ref{fig:privacy_utility_mechanism_level}, FO-DP-SGD
achieves a substantially better privacy--utility tradeoff than both standard
DP-SGD and the post-processing fractional-memory baseline. In particular, for
the same range of privacy budgets, the FO-DP-SGD curve remains above the
Post-FM-DP-SGD curve over most of the trajectory. This indicates that the
improvement is not merely due to applying a smoothing filter to already
privatized gradients. Instead, the result supports the central mechanism-level
claim that incorporating fractional memory into the sum-level private query
before Gaussian perturbation yields a more effective use of the privacy budget.

The comparison separates two effects. First, Post-FM-DP-SGD improves over
standard DP-SGD, showing that fractional memory itself can stabilize private
optimization. Second, FO-DP-SGD further improves over Post-FM-DP-SGD, showing
that the point at which memory is inserted matters. Thus, the observed gain is
associated not only with the presence of memory, but with its integration into
the private query mechanism before noise is added.

\subsection{Mechanism-Level Memory Improves Training Utility}
\label{subsec:results_mechanism_training_utility}

We next evaluate how the memory insertion point affects the training trajectory.
Figure~\ref{fig:mechanism_memory_vs_post_epoch} compares standard DP-SGD,
Post-FM-DP-SGD, and FO-DP-SGD in terms of test accuracy over training epochs.
The post-processing baseline applies fractional memory only after the standard
DP-SGD private release has already been generated, whereas FO-DP-SGD inserts the
memory term directly into the sum-level private query before Gaussian
perturbation.

\begin{figure}[H]
    \centering
    \includegraphics[width=0.82\linewidth]{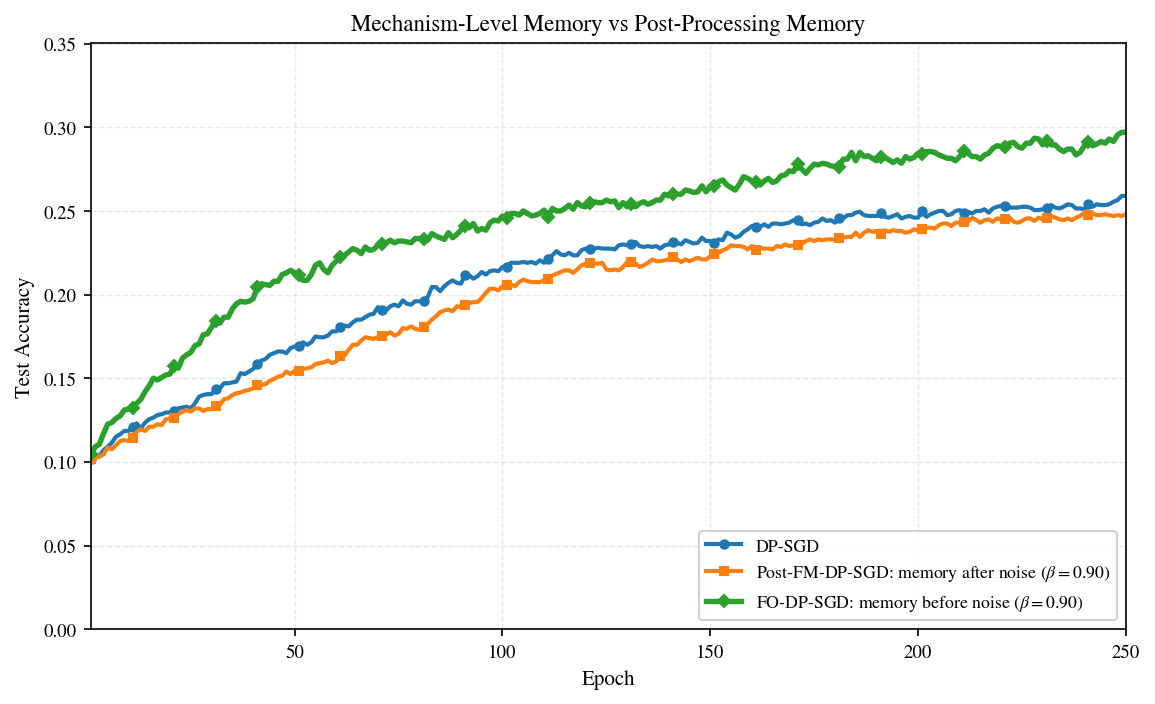}
    \caption{
    \textbf{Mechanism-level memory versus post-processing memory.}
    Test accuracy is reported over training epochs on CIFAR-10. DP-SGD is the
    standard no-memory private baseline. Post-FM-DP-SGD applies fractional
    memory after the noisy DP-SGD release, while FO-DP-SGD incorporates memory
    before Gaussian perturbation at the sum-query level. FO-DP-SGD achieves
    consistently higher test accuracy and faster improvement than both DP-SGD
    and Post-FM-DP-SGD, supporting the benefit of mechanism-level memory
    insertion.
    }
    \label{fig:mechanism_memory_vs_post_epoch}
\end{figure}

As shown in Figure~\ref{fig:mechanism_memory_vs_post_epoch}, FO-DP-SGD provides
a consistently stronger training trajectory than both baselines. Standard
DP-SGD improves gradually but remains below the two memory-based methods,
indicating that recursive memory can help stabilize noisy private optimization.
Post-FM-DP-SGD improves over DP-SGD, showing that applying fractional memory to
already privatized gradients is beneficial. However, FO-DP-SGD achieves the
highest accuracy throughout most of training, suggesting that inserting memory
before Gaussian perturbation at the query level is more effective than applying
the same memory rule only as a post-processing filter.

This result supports the mechanism-level claim of FO-DP-SGD: the gain is not
only due to the presence of memory, but also to where memory is introduced in
the private optimization pipeline. By modifying the sum-level private query
before noise is added, FO-DP-SGD obtains a more useful update direction and
improves the utility of private training.

\subsection{Runtime Comparison on CIFAR-10}
\label{subsec:runtime_comparison_cifar10}

We also compare the wall-clock training time of FO-DP-SGD against several
private optimization baselines on CIFAR-10. This comparison measures the total
runtime required to complete $250$ training epochs under the same CIFAR-10
experimental protocol. The goal is to assess whether the proposed
mechanism-level fractional memory introduces substantial computational overhead
relative to standard and adaptive DP optimizers.

\begin{table}[H]
\centering
\caption{
\textbf{Runtime comparison on CIFAR-10.}
Total wall-clock runtime is reported for $250$ training epochs on CIFAR-10.
Lower runtime is better. FO-DP-SGD with $\beta=1.00$ corresponds to the
no-memory limit, while FO-DP-SGD with $\beta=0.90$, $\alpha=0.80$, and
$K=8$ is the proposed memory-active configuration.
}
\label{tab:runtime_cifar10_private_optimizers}
\small
\begin{tabular}{lc}
\toprule
\textbf{Method} & \textbf{Runtime on CIFAR-10, 250 epochs (s)} \\
\midrule
FO-DP-SGD, $\beta=1.00$ & $59.1$ \\
DP-Adam & $76.1$ \\
DP-AdamW & $65.7$ \\
DP-IS & $65.6$ \\
DP-SAM & $92.6$ \\
DP-SAT & $88.6$ \\
DP-Adam-AC & $60.7$ \\
FO-DP-SGD, $\beta=0.90,\alpha=0.80,K=8$ & $60.1$ \\
\bottomrule
\end{tabular}
\end{table}

As shown in Table~\ref{tab:runtime_cifar10_private_optimizers}, the proposed
memory-active FO-DP-SGD configuration has a runtime of $60.1$ seconds on
CIFAR-10, which is only slightly higher than the no-memory FO-DP-SGD setting
with $\beta=1.00$, requiring $59.1$ seconds. This indicates that adding the
fractional memory state at the query level introduces only a small computational
overhead in this CIFAR-10 experiment.

The runtime of FO-DP-SGD is also competitive with the adaptive DP baselines.
For example, the memory-active FO-DP-SGD configuration is substantially faster
than DP-Adam, DP-SAM, and DP-SAT, which require $76.1$, $92.6$, and
$88.6$ seconds, respectively. It is also comparable to DP-Adam-AC, which
requires $60.7$ seconds. These results suggest that the proposed
mechanism-level memory does not impose a large runtime burden compared with
common adaptive private optimizers on CIFAR-10.

\subsection{Relative Runtime Overhead on CIFAR-10}
\label{subsec:relative_runtime_overhead_cifar10}

In addition to absolute wall-clock runtime, we report the relative runtime
overhead of each optimizer on CIFAR-10. The runtime is normalized with respect
to the no-memory FO-DP-SGD configuration, corresponding to $\beta=1.00$. Thus, a
value of $1.00\times$ indicates the reference runtime, while values larger than
$1.00\times$ indicate additional computational overhead.

\begin{table}[H]
\centering
\caption{
\textbf{Relative runtime overhead on CIFAR-10.}
Runtime is normalized by the no-memory FO-DP-SGD configuration
$(\beta=1.00)$, which serves as the $1.00\times$ reference. Lower values
indicate lower computational overhead. The memory-active FO-DP-SGD configuration
with $\beta=0.90$, $\alpha=0.80$, and $K=8$ introduces only a small relative
overhead.
}
\label{tab:relative_runtime_overhead_cifar10}
\small
\begin{tabular}{lc}
\toprule
\textbf{Method} & \textbf{Relative Runtime on CIFAR-10} \\
\midrule
FO-DP-SGD, $\beta=1.00$ & $1.00\times$ \\
DP-Adam & $1.29\times$ \\
DP-AdamW & $1.11\times$ \\
DP-IS & $1.11\times$ \\
DP-SAM & $1.57\times$ \\
DP-SAT & $1.50\times$ \\
DP-Adam-AC & $1.03\times$ \\
FO-DP-SGD, $\beta=0.90,\alpha=0.80,K=8$ & $1.02\times$ \\
\bottomrule
\end{tabular}
\end{table}

As shown in Table~\ref{tab:relative_runtime_overhead_cifar10}, the proposed
memory-active FO-DP-SGD configuration requires only $1.02\times$ the runtime of
the no-memory FO-DP-SGD reference. This indicates that incorporating fractional
memory into the private query introduces very limited additional computational
cost on CIFAR-10.

The relative runtime of FO-DP-SGD is also competitive compared with adaptive
private optimizers. DP-Adam requires $1.29\times$ the reference runtime, while
DP-SAM and DP-SAT require $1.57\times$ and $1.50\times$, respectively. Although
these methods introduce adaptive or sharpness-aware update mechanisms, their
computational overhead is substantially larger than that of the proposed
memory-active FO-DP-SGD. DP-Adam-AC is close to FO-DP-SGD in runtime, requiring
$1.03\times$, but FO-DP-SGD remains slightly more efficient.

\subsection{Effect of Clipping Norm under Fixed Noise}
\label{subsec:fixed_sigma_varying_clipping}

We next study how the clipping norm $C$ affects test accuracy when the noise
multiplier is fixed. In DP-SGD-type methods, the clipping norm controls the
maximum contribution of each per-example gradient, while the Gaussian noise is
scaled as $Z_t \sim \mathcal N(0,\sigma^2 C^2 I).$ Thus, changing $C$ affects the optimization dynamics by changing both the amount
of gradient clipping and the absolute noise scale. To isolate the effect of
clipping, we fix the noise multiplier and compare different clipping values
$C\in\{0.5,1.0,1.5,2.0\}$.

\begin{figure}[H]
    \centering
    \includegraphics[width=0.98\linewidth]{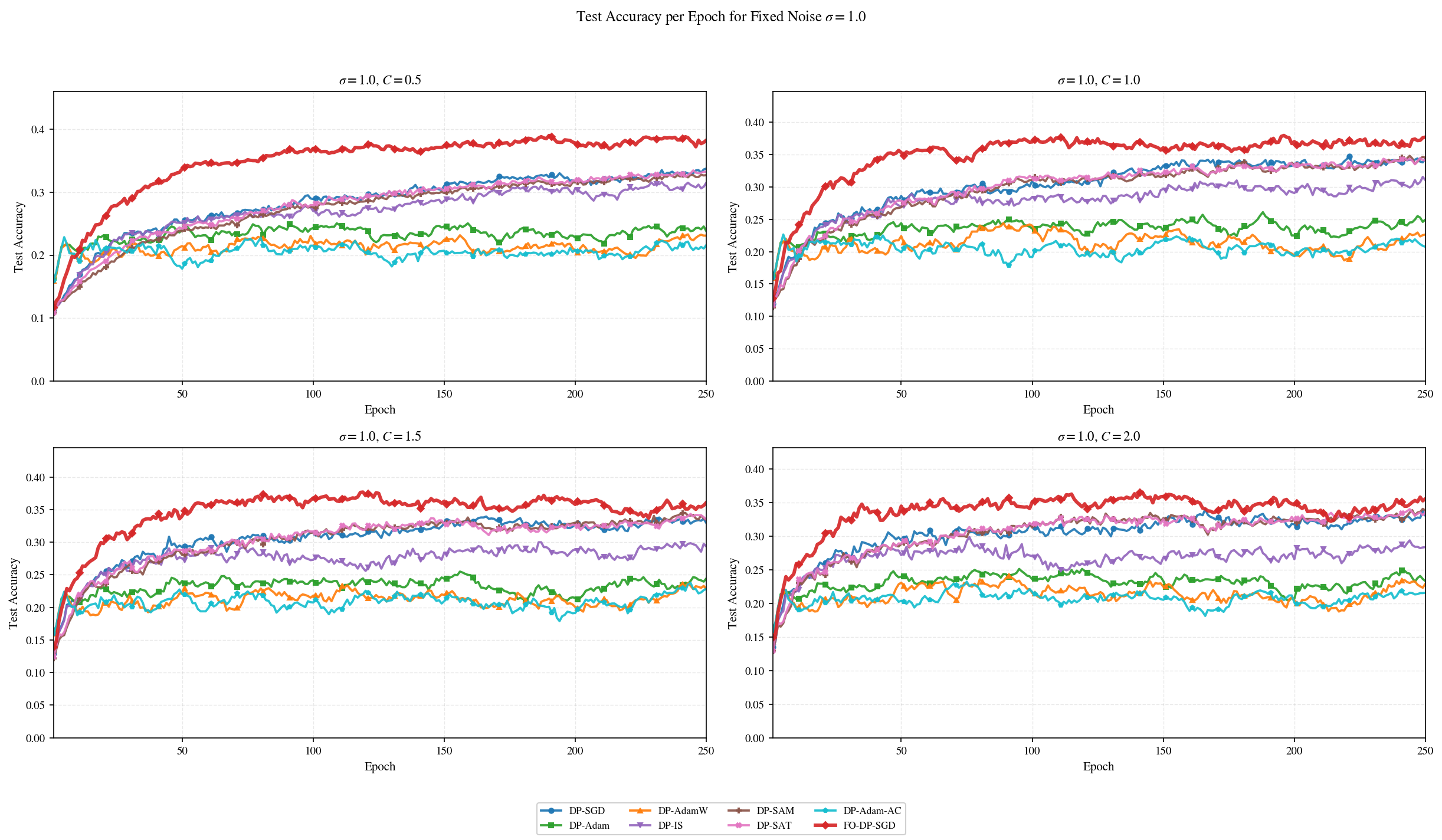}
    \caption{
    \textbf{Test accuracy under fixed noise $\sigma=1.5$ and varying clipping
    norms on CIFAR-10.}
    Test accuracy is reported over training epochs for
    $C\in\{0.5,1.0,1.5,2.0\}$ with fixed noise multiplier $\sigma=1.5$. Each
    subplot compares DP-SGD, adaptive DP optimizers, and FO-DP-SGD under the
    same fixed noise multiplier. FO-DP-SGD consistently achieves the strongest
    accuracy trajectory across all clipping values, indicating that the proposed
    fractional-memory mechanism remains robust to changes in the clipping norm.
    }
    \label{fig:fixed_sigma_1p5_varying_clipping}
\end{figure}

As shown in Figures~\ref{fig:fixed_sigma_1p5_varying_clipping}
and~\ref{fig:fixed_sigma_2p0_varying_clipping}, FO-DP-SGD maintains the highest
test accuracy across all tested clipping values under both fixed-noise settings.
This suggests that the proposed fractional-memory update is not effective only
for a single clipping configuration, but remains stable across a range of
clipping norms and noise levels. The advantage is especially visible throughout
the middle and late stages of training, where FO-DP-SGD consistently stays above
the competing private optimizers.

The results also show that changing $C$ affects the baselines differently. Some
adaptive methods exhibit relatively flat or unstable trajectories as the
clipping norm changes, suggesting sensitivity to the balance between clipping
bias and injected Gaussian noise. In contrast, FO-DP-SGD shows a more stable
accuracy trajectory across the four clipping regimes. This behavior supports
the interpretation that query-level fractional memory helps smooth noisy private
updates while preserving useful historical gradient information.

The comparison between $\sigma=1.5$ and $\sigma=2.0$ further shows that
FO-DP-SGD remains effective even when the noise multiplier is increased. Since
larger $\sigma$ introduces stronger perturbation, maintaining a consistent
accuracy advantage under $\sigma=2.0$ provides additional evidence that the
fractional-memory mechanism improves robustness to DP noise. Overall, these
results support the effectiveness of FO-DP-SGD under fixed-noise clipping
variation.

\subsection{ Accuracy Analysis}
\label{subsec:final-accuracy-analysis}

Table~\ref{tab:final-acc-ci95} reports the final test accuracy of all compared
methods over five independent random seeds. For each method, we report the mean
$\pm$ standard deviation and the corresponding 95\% confidence interval. This
evaluation allows us to assess the statistical reliability of the final accuracy
improvements under the same differentially private training protocol.

\begin{table}[H]
\centering
\caption{
Final accuracy comparison over five independent random seeds.
We report mean $\pm$ standard deviation and the 95\% confidence interval for
final test accuracy.
}
\label{tab:final-acc-ci95}
\resizebox{0.78\textwidth}{!}{
\begin{tabular}{lcc}
\toprule
\textbf{Method} 
& \textbf{Final Acc. mean $\pm$ std} 
& \textbf{Final Acc. 95\% CI} \\
\midrule
FO-DP-SGD $(\beta=0.90,\alpha=0.80,K=8)$ 
& $\mathbf{0.3654 \pm 0.0075}$ 
& $\mathbf{[0.3560,\;0.3748]}$ \\
DP-SAM 
& $0.3292 \pm 0.0111$ 
& $[0.3154,\;0.3430]$ \\
DP-SAT 
& $0.3281 \pm 0.0105$ 
& $[0.3150,\;0.3412]$ \\
FO-DP-SGD $(\beta=1.00)$ 
& $0.3241 \pm 0.0091$ 
& $[0.3128,\;0.3354]$ \\
DP-IS 
& $0.2928 \pm 0.0079$ 
& $[0.2829,\;0.3027]$ \\
DP-Adam 
& $0.2438 \pm 0.0108$ 
& $[0.2304,\;0.2572]$ \\
DP-AdamW 
& $0.2268 \pm 0.0103$ 
& $[0.2141,\;0.2395]$ \\
DP-Adam-AC 
& $0.2054 \pm 0.0086$ 
& $[0.1947,\;0.2161]$ \\
\bottomrule
\end{tabular}
}
\end{table}

As shown in Table~\ref{tab:final-acc-ci95}, the proposed FO-DP-SGD
configuration with $\beta=0.90$, $\alpha=0.80$, and $K=8$ achieves the highest
final test accuracy among all evaluated methods. It obtains a final accuracy of
$0.3654 \pm 0.0075$ with a 95\% confidence interval of $[0.3560, 0.3748]$. In
contrast, the non-memory FO-DP-SGD variant with $\beta=1.00$ achieves
$0.3241 \pm 0.0091$ with a 95\% confidence interval of $[0.3128, 0.3354]$. The
non-overlapping confidence intervals between these two variants indicate that
the improvement obtained by the proposed fractional-memory configuration is
statistically consistent across the evaluated random seeds.

Compared with the strongest differentially private optimizer baselines, the
proposed method also provides a clear improvement in final accuracy. DP-SAM
reaches $0.3292 \pm 0.0111$ with a 95\% confidence interval of
$[0.3154, 0.3430]$, while DP-SAT reaches $0.3281 \pm 0.0105$ with a 95\%
confidence interval of $[0.3150, 0.3412]$. Therefore, the proposed FO-DP-SGD
method improves final accuracy by approximately $3.62$ percentage points over
DP-SAM and $3.73$ percentage points over DP-SAT. This result suggests that
incorporating fractional-order memory into the private optimization dynamics
improves the final utility of the trained model under the same experimental
setting.

\subsection{Ablation on Query-Level Memory Design}
\label{subsec:ablation_query_memory_design}

To isolate the role of the query-level memory construction, we compare several
memory designs under the same private training protocol. All variants use the
same mechanism-level template: a memory-enhanced sum-level query is formed
before Gaussian perturbation, and then Gaussian noise is added at the sum level.
Thus, the comparison does not change the location at which memory enters the
private optimization pipeline; instead, it changes only how the historical
private releases are weighted when constructing the memory state.

At iteration \(t\), each memory-based variant forms a query of the form
\begin{equation}
r_t
=
\beta s_t(D;m_t)
+
(1-\beta)u_{t-1},
\label{eq:generic_query_memory_design}
\end{equation}
where \(s_t(D;m_t)\) is the current clipped subsampled sum and \(u_{t-1}\) is a
memory state constructed from previous private sum-level releases. After
forming \(r_t\), the mechanism releases
\begin{equation}
\tilde s_t
=
r_t+Z_t,
\qquad
Z_t\sim\mathcal N(0,\sigma^2 C^2 I),
\label{eq:query_memory_design_release}
\end{equation}
and the model is updated using \(\tilde s_t/L\). Therefore, all memory variants
remain mechanism-level query modifications rather than post-processing filters
applied after a standard DP-SGD release.

We compare four query-level memory designs.

\paragraph{Current-only / standard DP-SGD.}
The current-only variant removes the recursive memory term entirely. It is
equivalent to standard DP-SGD under the same clipping, sampling, Gaussian noise,
and privacy-accounting setup. Its query is
\begin{equation}
r_t^{\mathrm{cur}}
=
s_t(D;m_t),
\qquad
u_{t-1}^{\mathrm{cur}}=0.
\label{eq:current_only_memory_design}
\end{equation}
Equivalently, this corresponds to the generic query in
\eqref{eq:generic_query_memory_design} with \(\beta=1\) and no memory term.
This baseline tests whether recursive query-level memory is beneficial beyond
the ordinary current-gradient DP-SGD mechanism.

\paragraph{Uniform memory.}
The uniform-memory variant uses the same query-level template but constructs
the memory state as an equal-weight average of the previous private sum-level
releases:
\begin{equation}
u_{t-1}^{\mathrm{uni}}
=
\frac{1}{K_t-1}
\sum_{j=1}^{K_t-1}
\tilde s_{t-j},
\qquad
K_t=\min\{K,t+1\}.
\label{eq:uniform_memory_design}
\end{equation}
The corresponding query is
\begin{equation}
r_t^{\mathrm{uni}}
=
\beta s_t(D;m_t)
+
(1-\beta)u_{t-1}^{\mathrm{uni}}.
\label{eq:uniform_query_memory_design}
\end{equation}
This variant tests whether simply averaging prior private releases is sufficient
to improve private optimization.

\paragraph{Exponential memory.}
The exponential-memory variant assigns geometrically decaying weights to older
private releases. For a decay parameter \(\gamma_{\mathrm{exp}}\in(0,1)\), the
normalized weights are
\begin{equation}
w_{t,j}^{\mathrm{exp}}
=
\frac{\gamma_{\mathrm{exp}}^{j-1}}
{\sum_{\ell=1}^{K_t-1}\gamma_{\mathrm{exp}}^{\ell-1}},
\qquad j=1,\ldots,K_t-1.
\label{eq:exponential_weights_memory_design}
\end{equation}
The memory state is then
\begin{equation}
u_{t-1}^{\mathrm{exp}}
=
\sum_{j=1}^{K_t-1}
w_{t,j}^{\mathrm{exp}}\tilde s_{t-j},
\label{eq:exponential_memory_design}
\end{equation}
and the query becomes
\begin{equation}
r_t^{\mathrm{exp}}
=
\beta s_t(D;m_t)
+
(1-\beta)u_{t-1}^{\mathrm{exp}}.
\label{eq:exponential_query_memory_design}
\end{equation}
This variant tests whether a standard short-memory decay rule can match the
benefit of the proposed fractional query-level memory.

\paragraph{Confidence-aware tempered fractional memory: FO-DP-SGD.}
The proposed FO-DP-SGD variant uses the same query-level template in
\eqref{eq:generic_query_memory_design}, but constructs \(u_{t-1}\) using the
confidence-aware tempered fractional memory introduced in
Section~\ref{subsec:fodpsgd_ema_confidence}. In particular, FO-DP-SGD combines
a fractional power-law memory kernel with EMA-based inconsistency measurement
and confidence-aware tempering before forming the recursive sum-level query.
Thus, unlike uniform memory, which treats recent releases equally, and
exponential memory, which imposes a fixed geometric decay, FO-DP-SGD uses a
transcript-dependent fractional memory rule that can suppress stale or
inconsistent private releases while preserving useful long-memory information.

\begin{figure}[H]
    \centering
    \includegraphics[width=0.82\linewidth]{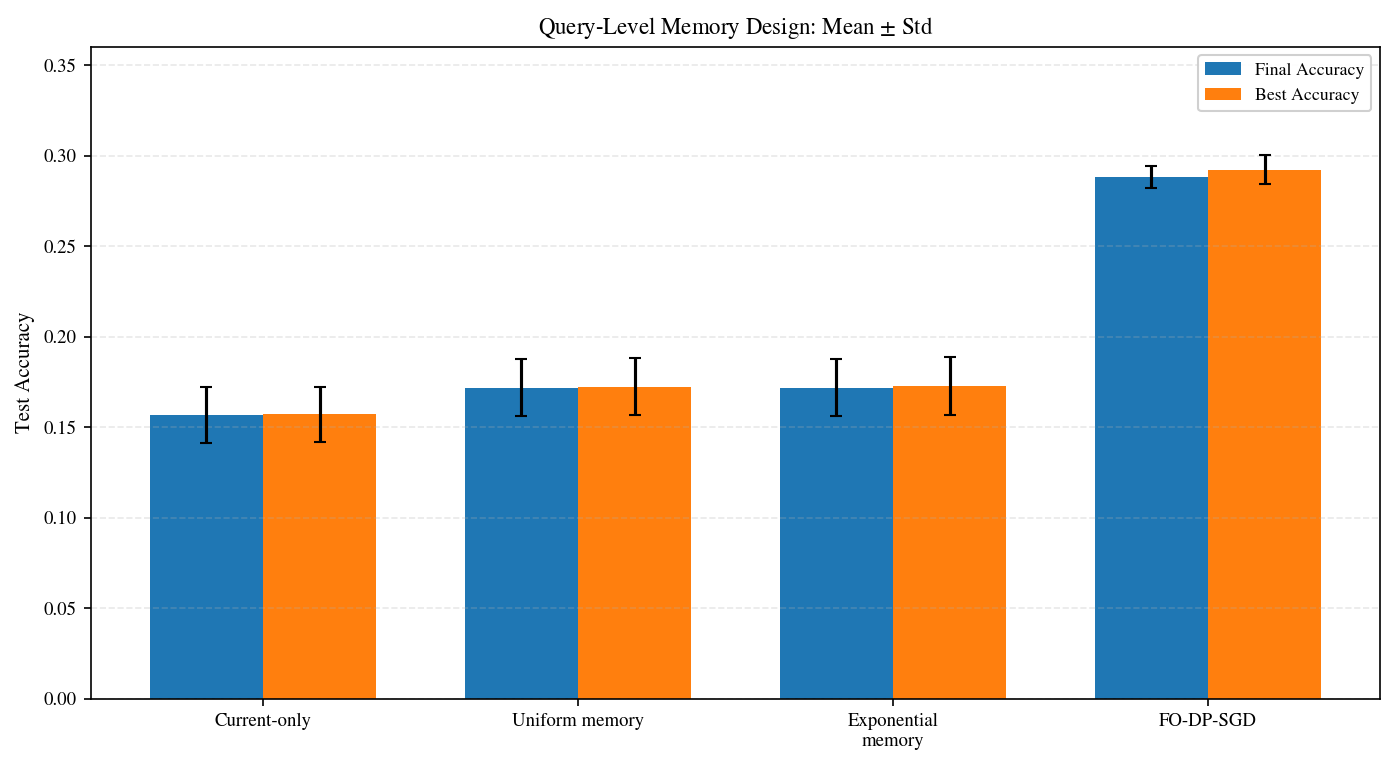}
    \caption{
    \textbf{Ablation on query-level memory design.}
    Final and best test accuracy are reported as mean \(\pm\) standard deviation
    over multiple runs. The current-only variant is equivalent to standard
    DP-SGD and removes recursive memory. Uniform memory uses equal weights over
    previous private releases, exponential memory uses a geometric decay kernel,
    and FO-DP-SGD uses the proposed confidence-aware tempered fractional
    query-level memory from Section~\ref{subsec:fodpsgd_ema_confidence}.
    FO-DP-SGD achieves the highest final and best accuracy, indicating that the
    observed gain is not merely due to adding memory, but also to the specific
    confidence-aware fractional structure of the query-level memory.
    }
    \label{fig:query_memory_design_ablation}
\end{figure}

Figure~\ref{fig:query_memory_design_ablation} shows that the proposed
FO-DP-SGD memory design substantially outperforms the other query-level memory
variants. The current-only baseline, which is equivalent to standard DP-SGD,
obtains the lowest accuracy, confirming that removing recursive memory weakens
private optimization under the evaluated setting. Uniform and exponential
memory improve over the current-only DP-SGD variant, suggesting that using
previous private releases can stabilize the noisy training trajectory. However,
both remain below the proposed FO-DP-SGD design.

The strongest performance is achieved by FO-DP-SGD, which uses the
confidence-aware tempered fractional memory introduced in
Section~\ref{subsec:fodpsgd_ema_confidence}. This result suggests that the
benefit of FO-DP-SGD is not simply caused by adding any memory term to the
private query. Rather, the form of the memory matters: uniform memory treats all
recent releases equally, exponential memory imposes a fixed geometric decay,
while FO-DP-SGD combines fractional long-memory structure with
transcript-dependent inconsistency suppression. Consequently, this ablation
supports the claim that the proposed query-level confidence-aware fractional
memory is a key component of the overall improvement.

\section{Theoretical Analysis and Interpretation of FO-DP-SGD}
\label{app:theoretical-analysis}
\label{sec:fodpsgd_theory}

This section provides a theoretical analysis of FO-DP-SGD and clarifies why the
proposed recursive query remains compatible with standard differentially
private accounting. We analyze the method from four complementary perspectives:
privacy validity and accounting, structural decomposition of the recursive
private release, interpretation of the EMA-based confidence-aware tempering
rule, and special or limiting regimes that recover simpler private optimization
mechanisms.

The central point is that FO-DP-SGD modifies the \emph{sum-level private query}
before Gaussian perturbation, while preserving the classical
sum-then-noise-then-divide structure of DP-SGD. At iteration \(t\), the recursive
query and release are
\begin{equation}
r_t(D;m_t,h_t)
=
\beta s_t(D;m_t)
+
(1-\beta)u_{t-1}^{\mathrm{CA}}(h_t),
\label{eq:theory_recursive_query_recall}
\end{equation}
and
\begin{equation}
\tilde s_t
=
r_t(D;m_t,h_t)+Z_t,
\qquad
Z_t\sim\mathcal N(0,\sigma^2 C^2 I).
\label{eq:theory_release_recall}
\end{equation}
Here \(s_t(D;m_t)\) is the current clipped subsampled sum, \(h_t\) is the prior
private transcript, \(u_{t-1}^{\mathrm{CA}}(h_t)\) is the confidence-aware
fractional private-memory state, \(C\) is the clipping norm, \(\sigma\) is the
noise multiplier, and \(\beta\in(0,1]\) controls the current-step contribution.

\subsection{Standing Assumptions and Notation}
\label{subsec:fodpsgd_theory_assumptions}

We first state the assumptions used throughout the theoretical analysis.

\paragraph{Adjacency.}
We use add/remove adjacency. Two datasets \(D\) and \(D'\) are adjacent,
denoted \(D\sim D'\), if one can be obtained from the other by adding or
removing one data point.

\paragraph{Per-example clipping.}
At iteration \(t\), every per-example gradient is clipped to have Euclidean norm
at most \(C\):
\begin{equation}
\|\bar g_t(x_i)\|_2\le C.
\label{eq:theory_clipping_assumption}
\end{equation}
For a Poisson sampling mask \(m_t=(m_{t,1},\ldots,m_{t,N})\), the clipped
subsampled sum is
\begin{equation}
s_t(D;m_t)
=
\sum_{i=1}^N m_{t,i}\bar g_t(x_i).
\label{eq:theory_subsampled_sum}
\end{equation}

\paragraph{Poisson subsampling.}
Each inclusion variable is sampled independently as
\[
m_{t,i}\sim \mathrm{Bernoulli}(q),
\]
where \(q\in(0,1]\) is the sampling probability. Privacy amplification is
accounted for using the Poisson-subsampled Gaussian mechanism.

\paragraph{Transcript-conditioned state.}
Let
\[
h_t=(\tilde s_0,\tilde s_1,\ldots,\tilde s_{t-1})
\]
denote the prior private transcript before the release at iteration \(t\). The
model parameter \(\theta_t\), the EMA trend, the inconsistency scores, the
confidence factor, the memory weights, and the memory state are deterministic
functions of the prior private transcript \(h_t\), public hyperparameters, and
the algorithmic update rule. Therefore, once \(h_t\) is fixed, the current
algorithmic state \(\theta_t\) and all transcript-dependent memory quantities
are fixed. The only fresh data-dependent quantity at iteration \(t\) is the
current clipped subsampled sum \(s_t(D;m_t)\).

\paragraph{Gaussian perturbation.}
FO-DP-SGD adds fresh independent Gaussian noise at the sum level:
\[
Z_t\sim \mathcal N(0,\sigma^2 C^2I).
\]
The released noisy gradient is \(\tilde g_t=\tilde s_t/L\), where \(L=Nq\) is
the expected lot size.

\subsection{Privacy Validity, Mechanism-Level Novelty, and Privacy Accounting}
\label{subsec:fodpsgd_theory_privacy_accounting}

We begin by formalizing the sense in which FO-DP-SGD remains privacy-valid while
modifying the sum-level query itself.

\begin{lemma}[Sensitivity of the clipped subsampled sum under a fixed mask]
\label{lem:fodpsgd_delta_s}
Fix the model state \(\theta_t\) and a realization of the sampling mask \(m_t\).
Under add/remove adjacency, the clipped subsampled sum satisfies
\begin{equation}
\sup_{D\sim D'}
\|s_t(D;m_t)-s_t(D';m_t)\|_2
\le C.
\label{eq:theory_clipped_sum_sensitivity}
\end{equation}
\end{lemma}

\begin{proof}
Under add/remove adjacency, the two datasets differ in at most one example. For
a fixed sampling mask \(m_t\), the two clipped subsampled sums can therefore
differ in at most one clipped per-example contribution. Since every clipped
gradient has norm at most \(C\), the difference between the two sums has norm at
most \(C\).
\end{proof}

\begin{proposition}[History-conditioned query structure]
\label{prop:fodpsgd_theory_history_conditioned}
Fix an iteration \(t\), a realizable prior private transcript \(h_t\), and a
realization of the current sampling mask \(m_t\). Since the current parameter
\(\theta_t\), EMA trend, inconsistency scores, confidence factor, normalized
memory weights, and private-memory state are deterministic functions of
\(h_t\) and public hyperparameters, they are fixed under this conditioning.
Consequently, the only newly data-dependent term in the recursive query
\(r_t(D;m_t,h_t)\) is \(\beta s_t(D;m_t)\).
\end{proposition}

\begin{proof}
The prior transcript \(h_t=(\tilde s_0,\ldots,\tilde s_{t-1})\) fixes all
previous private releases. Since the update rule is deterministic given the
prior transcript, public hyperparameters, and public algorithmic choices, the
current model state \(\theta_t\) is fixed under this conditioning. Therefore,
the clipped gradients \(\bar g_t(x_i)\) are evaluated at a fixed parameter
\(\theta_t\).

The EMA trend \(\bar s_t^{\mathrm{EMA}}\) is a deterministic function of prior
private releases. The inconsistency scores \(\nu_{t,j}\), the confidence factor
\(\chi_t\), the raw kernel coefficients \(a_{t,j}^{\mathrm{CA}}\), the
normalized weights \(\hat a_{t,j}^{\mathrm{CA}}\), and the memory state
\[
u_{t-1}^{\mathrm{CA}}(h_t)
=
\sum_{j=1}^{K_t-1}\hat a_{t,j}^{\mathrm{CA}}\tilde s_{t-j}
\]
are therefore also fixed. Hence, in
\[
r_t(D;m_t,h_t)
=
\beta s_t(D;m_t)
+
(1-\beta)u_{t-1}^{\mathrm{CA}}(h_t),
\]
the only newly data-dependent term is \(\beta s_t(D;m_t)\).
\end{proof}

\begin{definition}[History- and mask-conditioned recursive-query sensitivity]
\label{def:fodpsgd_theory_conditional_sensitivity}
For a fixed prior private transcript \(h_t\) and a fixed sampling mask \(m_t\),
define the recursive-query sensitivity
\begin{equation}
\Delta_r(m_t,h_t)
=
\sup_{D\sim D'}
\|r_t(D;m_t,h_t)-r_t(D';m_t,h_t)\|_2.
\label{eq:theory_recursive_sensitivity_def}
\end{equation}
\end{definition}

\begin{theorem}[Conditional sensitivity of the recursive query]
\label{thm:fodpsgd_theory_conditional_sensitivity}
Under add/remove adjacency, for every fixed prior private transcript \(h_t\)
and every fixed sampling mask \(m_t\),
\begin{equation}
\Delta_r(m_t,h_t)\le \beta C.
\label{eq:theory_recursive_sensitivity_bound}
\end{equation}
\end{theorem}

\begin{proof}
Fix \(h_t\) and \(m_t\). By
Proposition~\ref{prop:fodpsgd_theory_history_conditioned}, the memory term
\(u_{t-1}^{\mathrm{CA}}(h_t)\) is fixed. Hence, for any adjacent datasets
\(D\sim D'\),
\[
\begin{aligned}
&\|r_t(D;m_t,h_t)-r_t(D';m_t,h_t)\|_2 \\
&=
\left\|
\beta s_t(D;m_t)
+
(1-\beta)u_{t-1}^{\mathrm{CA}}(h_t)
-
\beta s_t(D';m_t)
-
(1-\beta)u_{t-1}^{\mathrm{CA}}(h_t)
\right\|_2 \\
&=
\beta\|s_t(D;m_t)-s_t(D';m_t)\|_2.
\end{aligned}
\]
By Lemma~\ref{lem:fodpsgd_delta_s}, the clipped subsampled sum has fixed-mask
sensitivity at most \(C\). Therefore,
\[
\Delta_r(m_t,h_t)\le \beta C.
\]
\end{proof}

\begin{definition}[Classical and FO-DP-SGD sum-level queries]
\label{def:fodpsgd_theory_queries}
At iteration \(t\), define the classical DP-SGD clipped-sum query as
\begin{equation}
q_t^{\mathrm{std}}(D;m_t)
=
s_t(D;m_t),
\label{eq:theory_standard_query}
\end{equation}
and the FO-DP-SGD recursive query as
\begin{equation}
q_t^{\mathrm{FO}}(D;m_t,h_t)
=
r_t(D;m_t,h_t)
=
\beta s_t(D;m_t)
+
(1-\beta)u_{t-1}^{\mathrm{CA}}(h_t).
\label{eq:theory_fo_query}
\end{equation}
\end{definition}

\begin{proposition}[Mechanism-level modification]
\label{prop:fodpsgd_theory_mechanism_level}
Suppose \(\beta<1\) and \(K_t\ge 2\). Then the FO-DP-SGD query
\(q_t^{\mathrm{FO}}(D;m_t,h_t)\) differs structurally from the classical
clipped-sum query \(q_t^{\mathrm{std}}(D;m_t)\). In particular, FO-DP-SGD is not
merely a post-processing transformation of a standard DP-SGD private release;
it modifies the sum-level query itself before the fresh Gaussian perturbation
at iteration \(t\) is applied.
\end{proposition}

\begin{proof}
When \(\beta<1\) and \(K_t\ge 2\), the recursive query contains the nonzero
history-dependent term
\[
(1-\beta)u_{t-1}^{\mathrm{CA}}(h_t).
\]
By contrast, the classical DP-SGD clipped-sum query is exactly
\[
q_t^{\mathrm{std}}(D;m_t)=s_t(D;m_t)
\]
and contains no contribution from previous private releases. Therefore,
\(q_t^{\mathrm{FO}}\) and \(q_t^{\mathrm{std}}\) are structurally different
query objects.

Moreover, the fresh Gaussian noise in FO-DP-SGD is applied to the recursive
query:
\[
\tilde s_t=q_t^{\mathrm{FO}}(D;m_t,h_t)+Z_t.
\]
Thus, FO-DP-SGD first constructs a new history-dependent sum-level query and
then privatizes that query. A purely post-processing construction would first
release a standard private DP-SGD quantity and only afterward transform the
already-private output. Hence, FO-DP-SGD is a mechanism-level modification.
\end{proof}

\begin{theorem}[Per-step RDP accounting]
\label{thm:fodpsgd_theory_per_step_privacy}
Let
\[
\varepsilon_{\mathrm{SGM}}(\lambda;q,\rho)
\]
denote any valid R\'enyi differential privacy upper bound, at order
\(\lambda>1\), for the Poisson-subsampled Gaussian mechanism with subsampling
probability \(q\) and noise-to-sensitivity ratio \(\rho\). Conditioned on any
realizable prior private transcript \(h_t\), the FO-DP-SGD release
\[
\tilde s_t=r_t(D;m_t,h_t)+Z_t
\]
satisfies
\begin{equation}
\left(\lambda,
\varepsilon_{\mathrm{SGM}}(\lambda;q,\sigma/\beta)
\right)\text{-RDP}.
\label{eq:theory_per_step_rdp}
\end{equation}
\end{theorem}

\begin{proof}
Conditioned on the prior transcript \(h_t\), the memory term is fixed and the
only fresh data-dependent quantity is the Poisson-subsampled clipped sum
scaled by \(\beta\). By Theorem~\ref{thm:fodpsgd_theory_conditional_sensitivity},
the fixed-mask sensitivity of the recursive query is at most \(\beta C\). The
added Gaussian noise has standard deviation \(\sigma C\) at the sum level.
Therefore, the effective noise-to-sensitivity ratio is
\[
\frac{\sigma C}{\beta C}
=
\frac{\sigma}{\beta}.
\]
Since the current data access is performed through Poisson subsampling with
probability \(q\), the release is accounted for using the Poisson-subsampled
Gaussian mechanism with ratio \(\sigma/\beta\). This gives the stated
per-step RDP bound.
\end{proof}

\begin{theorem}[Adaptive composition over \(T\) iterations]
\label{thm:fodpsgd_theory_total_privacy}
Suppose FO-DP-SGD is run for \(T\) iterations with fixed
\((q,\beta,\sigma)\). Then, for every R\'enyi order \(\lambda>1\), the transcript
\[
\mathcal T(D)
=
(\tilde s_0,\tilde s_1,\ldots,\tilde s_{T-1})
\]
satisfies
\begin{equation}
\left(\lambda,
\varepsilon_{\mathrm{tot}}(\lambda)
\right)\text{-RDP},
\qquad
\varepsilon_{\mathrm{tot}}(\lambda)
=
T\,\varepsilon_{\mathrm{SGM}}(\lambda;q,\sigma/\beta).
\label{eq:theory_total_rdp}
\end{equation}
Consequently, for any \(\delta\in(0,1)\), the mechanism satisfies
\((\varepsilon_\delta,\delta)\)-DP with
\begin{equation}
\varepsilon_\delta
=
\inf_{\lambda>1}
\left\{
T\,\varepsilon_{\mathrm{SGM}}(\lambda;q,\sigma/\beta)
+
\frac{\log(1/\delta)}{\lambda-1}
\right\}.
\label{eq:theory_rdp_to_dp}
\end{equation}
If \((q_t,\beta_t,\sigma_t)\) vary across iterations, the more general
composition bound is
\begin{equation}
\varepsilon_{\mathrm{tot}}(\lambda)
=
\sum_{t=0}^{T-1}
\varepsilon_{\mathrm{SGM}}(\lambda;q_t,\sigma_t/\beta_t).
\label{eq:theory_variable_total_rdp}
\end{equation}
\end{theorem}

\begin{proof}
The per-step guarantee follows from
Theorem~\ref{thm:fodpsgd_theory_per_step_privacy}. Although the algorithm is
adaptive, each release may depend on previous private releases only through the
prior private transcript and deterministic post-processing of that transcript.
Adaptive sequential composition of RDP therefore applies. Summing the per-step
RDP costs gives
\[
\varepsilon_{\mathrm{tot}}(\lambda)
=
\sum_{t=0}^{T-1}\varepsilon_t(\lambda).
\]
For fixed \((q,\beta,\sigma)\), each term equals
\(\varepsilon_{\mathrm{SGM}}(\lambda;q,\sigma/\beta)\), giving
\eqref{eq:theory_total_rdp}. The conversion from RDP to
\((\varepsilon,\delta)\)-DP yields \eqref{eq:theory_rdp_to_dp}. The
time-varying case follows by summing the corresponding per-step costs.
\end{proof}

\begin{remark}[Why \(\sigma/\beta\) appears]
\label{rem:fodpsgd_theory_sigma_over_beta}
The ratio \(\sigma/\beta\) is not an additional modeling assumption. It follows
from the recursive-query sensitivity bound and the sum-level Gaussian
perturbation. The current-step sensitivity is bounded by \(\beta C\), while the
Gaussian noise has standard deviation \(\sigma C\). Hence the effective
noise-to-sensitivity ratio is \(\sigma/\beta\).
\end{remark}

\begin{remark}[Conditioning and privacy amplification]
\label{rem:fodpsgd_theory_conditioning}
The transcript conditioning isolates the newly data-dependent component of the
current recursive query. The fixed-mask sensitivity argument characterizes the
sensitivity of the clipped subsampled sum for each realized mask. The final
privacy accountant does not condition away the subsampling randomness; instead,
it retains Poisson subsampling as part of the mechanism and uses the
subsampled-Gaussian RDP bound.
\end{remark}

\subsection{Signal--Memory--Noise Decomposition}
\label{subsec:fodpsgd_theory_decomposition}

We next decompose the recursive release into current signal, inherited
recursive-state memory, inherited release-noise memory, and fresh Gaussian
noise. This decomposition is useful for understanding both the optimization
benefit and the limitations of the method.

\begin{definition}[Inherited recursive-state memory and inherited release-noise memory]
\label{def:fodpsgd_theory_memory_terms}
For \(K_t\ge 2\), define
\begin{align}
M_t^{\mathrm{rec}}
&=
(1-\beta)\sum_{j=1}^{K_t-1}\hat a_{t,j}^{\mathrm{CA}}\,
r_{t-j}(D;m_{t-j},h_{t-j}),
\label{eq:theory_m_rec}
\\
M_t^{\mathrm{noise}}
&=
(1-\beta)\sum_{j=1}^{K_t-1}\hat a_{t,j}^{\mathrm{CA}}\,Z_{t-j},
\label{eq:theory_m_noise}
\end{align}
where each previous release satisfies
\[
\tilde s_{t-j}
=
r_{t-j}(D;m_{t-j},h_{t-j})+Z_{t-j}.
\]
If \(K_t=1\), both \(M_t^{\mathrm{rec}}\) and \(M_t^{\mathrm{noise}}\) are
defined to be zero.
\end{definition}

\begin{lemma}[Signal--memory--noise decomposition]
\label{lem:fodpsgd_theory_decomposition}
At every iteration \(t\), the recursive query, private release, and noisy
gradient admit the decompositions
\begin{align}
r_t(D;m_t,h_t)
&=
\beta s_t(D;m_t)
+
M_t^{\mathrm{rec}}
+
M_t^{\mathrm{noise}},
\label{eq:fodpsgd_theory_rt_decomp_new}
\\
\tilde s_t
&=
\beta s_t(D;m_t)
+
M_t^{\mathrm{rec}}
+
M_t^{\mathrm{noise}}
+
Z_t,
\label{eq:fodpsgd_theory_tildes_decomp_new}
\\
\tilde g_t
&=
\frac{1}{L}
\left(
\beta s_t(D;m_t)
+
M_t^{\mathrm{rec}}
+
M_t^{\mathrm{noise}}
+
Z_t
\right).
\label{eq:fodpsgd_theory_tildeg_decomp_new}
\end{align}
\end{lemma}

\begin{proof}
Starting from the recursive query,
\[
r_t(D;m_t,h_t)
=
\beta s_t(D;m_t)
+
(1-\beta)\sum_{j=1}^{K_t-1}\hat a_{t,j}^{\mathrm{CA}}\tilde s_{t-j},
\]
substitute
\[
\tilde s_{t-j}
=
r_{t-j}(D;m_{t-j},h_{t-j})+Z_{t-j}
\]
for each active lag \(j\). Then
\[
\begin{aligned}
r_t(D;m_t,h_t)
&=
\beta s_t(D;m_t)
+
(1-\beta)\sum_{j=1}^{K_t-1}\hat a_{t,j}^{\mathrm{CA}}
\bigl(r_{t-j}(D;m_{t-j},h_{t-j})+Z_{t-j}\bigr) \\
&=
\beta s_t(D;m_t)
+
(1-\beta)\sum_{j=1}^{K_t-1}\hat a_{t,j}^{\mathrm{CA}}
r_{t-j}(D;m_{t-j},h_{t-j}) \\
&\qquad
+
(1-\beta)\sum_{j=1}^{K_t-1}\hat a_{t,j}^{\mathrm{CA}}Z_{t-j}.
\end{aligned}
\]
By Definition~\ref{def:fodpsgd_theory_memory_terms}, the second and third terms
are \(M_t^{\mathrm{rec}}\) and \(M_t^{\mathrm{noise}}\). This proves
\eqref{eq:fodpsgd_theory_rt_decomp_new}. Adding the fresh Gaussian noise \(Z_t\)
gives \eqref{eq:fodpsgd_theory_tildes_decomp_new}, and dividing by \(L\) gives
\eqref{eq:fodpsgd_theory_tildeg_decomp_new}.
\end{proof}

\begin{proposition}[Convexity of the memory weights]
\label{prop:fodpsgd_theory_convex_weights}
For every iteration \(t\) with \(K_t\ge 2\), the normalized confidence-aware
fractional weights satisfy
\[
\hat a_{t,j}^{\mathrm{CA}}\ge 0,
\qquad
\sum_{j=1}^{K_t-1}\hat a_{t,j}^{\mathrm{CA}}=1.
\]
Hence the private memory state
\[
u_{t-1}^{\mathrm{CA}}(h_t)
=
\sum_{j=1}^{K_t-1}\hat a_{t,j}^{\mathrm{CA}}\tilde s_{t-j}
\]
is a convex combination of previously released private sums.
\end{proposition}

\begin{proof}
The raw weights are
\[
a_{t,j}^{\mathrm{CA}}
=
(j+1)^{\alpha-1}
\exp\!\Bigl(-(\lambda+\chi_t\tau\nu_{t,j})j\Bigr).
\]
Since \(\alpha\in(0,1]\), \(j\ge 1\), \(\lambda\ge0\), \(\chi_t\ge0\),
\(\tau\ge0\), and \(\nu_{t,j}\ge0\), every raw coefficient is positive. The
normalized coefficient is
\[
\hat a_{t,j}^{\mathrm{CA}}
=
\frac{a_{t,j}^{\mathrm{CA}}}{\sum_{\ell=1}^{K_t-1}a_{t,\ell}^{\mathrm{CA}}},
\]
so each normalized coefficient is nonnegative and the normalized coefficients
sum to one.
\end{proof}

\begin{remark}[Inherited noise is not explicitly denoised]
\label{rem:fodpsgd_inherited_noise_no_denoising}
The decomposition in Lemma~\ref{lem:fodpsgd_theory_decomposition} shows that
the memory state contains both inherited recursive-state signal and inherited
Gaussian release noise. FO-DP-SGD can down-weight unreliable history through
inconsistency-aware and confidence-aware weighting, but it does not explicitly
separate or remove the noise terms \(Z_{t-j}\) from previous releases. Designing
transcript-based denoising or uncertainty-aware filtering rules for the private
memory state is therefore a natural direction for future work.
\end{remark}

\subsection{Interpretation of EMA-Based Confidence-Aware Tempering}
\label{subsec:fodpsgd_theory_interpretation}

We now analyze how the EMA-based confidence-aware tempering rule shapes the
private memory mechanism. The goal of this subsection is structural
interpretation rather than privacy accounting, which was established above.

The raw confidence-aware fractional kernel is
\begin{equation}
a_{t,j}^{\mathrm{CA}}
=
(j+1)^{\alpha-1}
\exp\!\Bigl(-(\lambda+\chi_t\tau\nu_{t,j})j\Bigr),
\qquad j=1,\ldots,K_t-1.
\label{eq:theory_ca_kernel}
\end{equation}

\begin{proposition}[Dependence on the inconsistency-aware tempering coefficient \(\tau\)]
\label{prop:fodpsgd_theory_tau_monotonicity}
Fix \(t\), \(j\), \(\alpha\), \(\lambda\), \(\chi_t\), and \(\nu_{t,j}\). The raw
coefficient \(a_{t,j}^{\mathrm{CA}}\) is nonincreasing in \(\tau\). If
\(\chi_t>0\) and \(\nu_{t,j}>0\), then it decreases strictly as \(\tau\)
increases.
\end{proposition}

\begin{proof}
The dependence on \(\tau\) appears only through
\[
\exp\!\Bigl(-(\lambda+\chi_t\tau\nu_{t,j})j\Bigr).
\]
Since \(j\ge1\), \(\chi_t\ge0\), and \(\nu_{t,j}\ge0\), increasing \(\tau\) can
only make the exponent more negative. If \(\chi_t>0\) and \(\nu_{t,j}>0\), the
exponent decreases strictly.
\end{proof}

\begin{proposition}[Suppression of inconsistent historical releases]
\label{prop:fodpsgd_theory_inconsistency_suppression}
Fix \(t\), \(j\), \(\alpha\), \(\lambda\), \(\chi_t\), and \(\tau\). The raw
coefficient \(a_{t,j}^{\mathrm{CA}}\) is nonincreasing in the inconsistency
score \(\nu_{t,j}\). If \(\chi_t>0\) and \(\tau>0\), then the decrease is
strict.
\end{proposition}

\begin{proof}
The dependence on \(\nu_{t,j}\) appears only through
\[
\exp\!\Bigl(-(\lambda+\chi_t\tau\nu_{t,j})j\Bigr).
\]
Since \(j\ge1\), \(\chi_t\ge0\), and \(\tau\ge0\), increasing \(\nu_{t,j}\) can
only decrease the exponential factor. The decrease is strict when
\(\chi_t>0\) and \(\tau>0\).
\end{proof}

\begin{proposition}[Dependence on the confidence factor \(\chi_t\)]
\label{prop:fodpsgd_theory_chi_monotonicity}
Fix \(t\), \(j\), \(\alpha\), \(\lambda\), \(\tau\), and \(\nu_{t,j}\). The raw
coefficient \(a_{t,j}^{\mathrm{CA}}\) is nonincreasing in \(\chi_t\). If
\(\tau>0\) and \(\nu_{t,j}>0\), then the decrease is strict.
\end{proposition}

\begin{proof}
The dependence on \(\chi_t\) appears only through
\[
\exp\!\Bigl(-(\lambda+\chi_t\tau\nu_{t,j})j\Bigr).
\]
Since \(j\ge1\), \(\tau\ge0\), and \(\nu_{t,j}\ge0\), increasing \(\chi_t\) can
only make the exponent smaller. The decrease is strict when \(\tau>0\) and
\(\nu_{t,j}>0\).
\end{proof}

Together, Propositions~\ref{prop:fodpsgd_theory_tau_monotonicity}--
\ref{prop:fodpsgd_theory_chi_monotonicity} characterize the role of
confidence-aware tempering. The inconsistency score \(\nu_{t,j}\) measures how
far a lagged private release is from the EMA-based private trend. The parameter
\(\tau\) controls the strength of inconsistency-aware attenuation. The
confidence factor \(\chi_t\) gates this attenuation according to the magnitude
of the private trend: when the trend is weak, the mechanism avoids overly
aggressive reweighting; when the trend is stronger, inconsistent history is
suppressed more strongly.

\begin{proposition}[Effect of \(\beta\) on current sensitivity and current signal]
\label{prop:fodpsgd_theory_beta_tradeoff}
Decreasing \(\beta\) reduces the current-step conditional sensitivity bound
linearly from \(C\) toward \(0\), but it also attenuates the direct coefficient
of the current clipped-sum contribution \(\beta s_t(D;m_t)\) by the same factor.
\end{proposition}

\begin{proof}
By Theorem~\ref{thm:fodpsgd_theory_conditional_sensitivity},
\[
\Delta_r(m_t,h_t)\le \beta C,
\]
which is linear in \(\beta\). Thus, smaller \(\beta\) gives a smaller
current-step sensitivity bound. At the same time, the recursive query is
\[
r_t(D;m_t,h_t)
=
\beta s_t(D;m_t)
+
(1-\beta)u_{t-1}^{\mathrm{CA}}(h_t),
\]
so the direct coefficient of the current clipped sum is exactly \(\beta\).
Reducing \(\beta\) therefore weakens the immediate current-gradient signal.
\end{proof}

\begin{remark}[Role of \(\lambda\)]
\label{rem:fodpsgd_theory_lambda}
For fixed \(\alpha\), \(\tau\), \(\chi_t\), and \(\nu_{t,j}\), the parameter
\(\lambda\) applies uniform exponential damping \(e^{-\lambda j}\) across
active lags. It controls baseline temporal forgetting independently of whether
a historical release is consistent or inconsistent with the recent private
trend.
\end{remark}

\begin{remark}[Role of \(\gamma\)]
\label{rem:fodpsgd_theory_gamma}
The EMA parameter \(\gamma\in(0,1]\) controls how quickly the private trend
\[
\bar s_t^{\mathrm{EMA}}
=
\gamma\tilde s_{t-1}
+
(1-\gamma)\bar s_{t-1}^{\mathrm{EMA}}
\]
reacts to recent private history. Larger \(\gamma\) makes the trend more
responsive to the latest release, while smaller \(\gamma\) yields a smoother and
slower-moving trend proxy.
\end{remark}

\begin{remark}[Role of \(\kappa\)]
\label{rem:fodpsgd_theory_kappa}
The parameter \(\kappa>0\) stabilizes the normalized inconsistency score
\[
\nu_{t,j}
=
\frac{\|\tilde s_{t-j}-\bar s_t^{\mathrm{EMA}}\|_2}
{\max(\|\bar s_t^{\mathrm{EMA}}\|_2,\kappa)+\varepsilon}
\]
by preventing the normalization scale from collapsing when the EMA trend
magnitude is small.
\end{remark}

\begin{remark}[Role of \(\zeta\)]
\label{rem:fodpsgd_theory_zeta}
The confidence parameter \(\zeta>0\) determines how aggressively the confidence
factor
\[
\chi_t
=
\frac{\|\bar s_t^{\mathrm{EMA}}\|_2}
{\|\bar s_t^{\mathrm{EMA}}\|_2+\zeta}
\]
responds to the trend magnitude. Larger \(\zeta\) yields smaller \(\chi_t\) for
a fixed trend norm and therefore weakens inconsistency-aware tempering; smaller
\(\zeta\) makes the transition toward stronger gating more rapid.
\end{remark}

\begin{remark}[Role of \(\alpha\)]
\label{rem:fodpsgd_theory_alpha}
The parameter \(\alpha\in(0,1]\) controls the fractional power-law component
\((j+1)^{\alpha-1}\). For \(\alpha\in(0,1)\), this factor decreases with the lag
\(j\), favoring more recent private releases. At \(\alpha=1\), the power-law
factor becomes constant, so lag discrimination is governed by exponential
tempering and confidence-aware inconsistency modulation.
\end{remark}

\begin{remark}[Role of \(K\)]
\label{rem:fodpsgd_theory_K}
The memory window \(K\) controls how many past private releases are eligible to
enter the recursive memory state through
\[
K_t=\min\{K,t+1\}.
\]
Larger \(K\) expands the support of the memory kernel and therefore broadens
both inherited recursive-state memory and inherited release-noise memory.
\end{remark}

\subsection{Special Cases and Limiting Regimes}
\label{subsec:fodpsgd_theory_special_cases}

We conclude with exact regimes that clarify how FO-DP-SGD relates to simpler
mechanisms.

\begin{proposition}[Case \(\tau=0\)]
\label{prop:fodpsgd_theory_tau_zero}
If \(\tau=0\), then
\[
a_{t,j}^{\mathrm{CA}}
=
(j+1)^{\alpha-1}e^{-\lambda j}
\]
for every active lag \(j\). Thus FO-DP-SGD reduces to a tempered fractional
private-memory mechanism without inconsistency-aware modulation.
\end{proposition}

\begin{proof}
Set \(\tau=0\) in
\[
a_{t,j}^{\mathrm{CA}}
=
(j+1)^{\alpha-1}
\exp\!\Bigl(-(\lambda+\chi_t\tau\nu_{t,j})j\Bigr).
\]
The inconsistency-aware term vanishes, leaving
\[
a_{t,j}^{\mathrm{CA}}
=
(j+1)^{\alpha-1}e^{-\lambda j}.
\]
\end{proof}

\begin{proposition}[Case \(\lambda=0\) and \(\tau=0\)]
\label{prop:fodpsgd_theory_lambda_tau_zero}
If \(\lambda=0\) and \(\tau=0\), then
\[
a_{t,j}^{\mathrm{CA}}
=
(j+1)^{\alpha-1}
\]
for every active lag \(j\). Hence FO-DP-SGD reduces to the raw power-law
fractional private-memory rule.
\end{proposition}

\begin{proof}
Substituting \(\lambda=0\) and \(\tau=0\) into the kernel makes the exponential
factor equal to one. Thus
\[
a_{t,j}^{\mathrm{CA}}
=
(j+1)^{\alpha-1}.
\]
\end{proof}

\begin{proposition}[Case \(\beta=1\)]
\label{prop:fodpsgd_theory_beta_one}
If \(\beta=1\), then for every iteration \(t\),
\[
r_t(D;m_t,h_t)=s_t(D;m_t),
\qquad
\tilde s_t=s_t(D;m_t)+Z_t,
\qquad
\tilde g_t=\frac{s_t(D;m_t)+Z_t}{L}.
\]
Thus FO-DP-SGD reduces exactly to the standard clipped-sum Gaussian release
under Poisson subsampling, i.e., ordinary DP-SGD under the same clipping,
sampling, noise, and privacy-accounting setup.
\end{proposition}

\begin{proof}
Setting \(\beta=1\) in
\[
r_t(D;m_t,h_t)
=
\beta s_t(D;m_t)
+
(1-\beta)u_{t-1}^{\mathrm{CA}}(h_t)
\]
eliminates the memory term and gives \(r_t(D;m_t,h_t)=s_t(D;m_t)\). Substituting
this into the release rule and dividing by \(L\) yields the stated expressions.
\end{proof}

\begin{proposition}[Case \(K=1\)]
\label{prop:fodpsgd_theory_K_one}
If \(K=1\), then \(K_t=1\) for all \(t\), the memory sum is empty, and
\[
r_t(D;m_t,h_t)=\beta s_t(D;m_t).
\]
Thus FO-DP-SGD becomes a purely current-step sensitivity-scaled release with no
temporal memory. It coincides with ordinary DP-SGD only in the additional
special case \(\beta=1\).
\end{proposition}

\begin{proof}
Since \(K_t=\min\{K,t+1\}\), setting \(K=1\) gives \(K_t=1\) for all \(t\). The
memory sum is empty, so \(u_{t-1}^{\mathrm{CA}}(h_t)=0\). Substitution into the
recursive query yields
\[
r_t(D;m_t,h_t)=\beta s_t(D;m_t).
\]
If additionally \(\beta=1\), this becomes the standard DP-SGD clipped-sum query.
\end{proof}

\begin{proposition}[Case \(\alpha=1\)]
\label{prop:fodpsgd_theory_alpha_one}
If \(\alpha=1\), then
\[
a_{t,j}^{\mathrm{CA}}
=
\exp\!\Bigl(-(\lambda+\chi_t\tau\nu_{t,j})j\Bigr).
\]
Thus the power-law component disappears and lag dependence is governed entirely
by baseline exponential tempering and confidence-aware inconsistency modulation.
\end{proposition}

\begin{proof}
Substitute \(\alpha=1\) into
\[
a_{t,j}^{\mathrm{CA}}
=
(j+1)^{\alpha-1}
\exp\!\Bigl(-(\lambda+\chi_t\tau\nu_{t,j})j\Bigr).
\]
Since \((j+1)^0=1\), the result follows.
\end{proof}

\begin{proposition}[Confidence-vanishing regime \(\chi_t\to0\)]
\label{prop:fodpsgd_theory_chi_to_zero}
For fixed \(t,j,\alpha,\lambda,\tau\), and \(\nu_{t,j}\), if
\(\chi_t\to0\), then
\[
a_{t,j}^{\mathrm{CA}}
\to
(j+1)^{\alpha-1}e^{-\lambda j}.
\]
Hence the mechanism approaches baseline tempered fractional memory without
effective inconsistency-aware tempering.
\end{proposition}

\begin{proof}
Taking \(\chi_t\to0\) in
\[
a_{t,j}^{\mathrm{CA}}
=
(j+1)^{\alpha-1}
\exp\!\Bigl(-(\lambda+\chi_t\tau\nu_{t,j})j\Bigr)
\]
removes the inconsistency-aware term and gives the stated limit.
\end{proof}

\begin{proposition}[Full-confidence regime \(\chi_t\to1\)]
\label{prop:fodpsgd_theory_chi_to_one}
For fixed \(t,j,\alpha,\lambda,\tau\), and \(\nu_{t,j}\), if
\(\chi_t\to1\), then
\[
a_{t,j}^{\mathrm{CA}}
\to
(j+1)^{\alpha-1}
\exp\!\bigl(-(\lambda+\tau\nu_{t,j})j\bigr).
\]
Thus the mechanism approaches fully active inconsistency-aware tempering without
confidence suppression.
\end{proposition}

\begin{proof}
Take the limit \(\chi_t\to1\) in the confidence-aware kernel. The result follows
immediately.
\end{proof}

\begin{proposition}[Memory-dominated regime \(\beta\to0\)]
\label{prop:fodpsgd_theory_beta_to_zero}
For fixed \(t\) and fixed prior private transcript \(h_t\), if \(\beta\to0\),
then
\[
r_t(D;m_t,h_t)\to u_{t-1}^{\mathrm{CA}}(h_t).
\]
Thus the current release becomes dominated by inherited private history rather
than by the current clipped-sum term.
\end{proposition}

\begin{proof}
Starting from
\[
r_t(D;m_t,h_t)
=
\beta s_t(D;m_t)
+
(1-\beta)u_{t-1}^{\mathrm{CA}}(h_t),
\]
taking \(\beta\to0\) eliminates the current clipped-sum term and leaves
\(u_{t-1}^{\mathrm{CA}}(h_t)\).
\end{proof}

\begin{proposition}[Full-participation limit \(q\to1\)]
\label{prop:fodpsgd_theory_q_to_one}
As \(q\to1\), Poisson subsampling approaches full participation, and the
clipped subsampled sum converges to the full clipped sum:
\[
s_t(D;m_t)
\to
\sum_{i=1}^N \bar g_t(x_i).
\]
Consequently, the recursive query approaches its full-participation
counterpart.
\end{proposition}

\begin{proof}
Under Poisson subsampling,
\[
m_{t,i}\sim\mathrm{Bernoulli}(q).
\]
As \(q\to1\), \(\mathbb P(m_{t,i}=1)\to1\), so \(m_{t,i}\to1\) almost surely for
each \(i\). Therefore,
\[
s_t(D;m_t)
=
\sum_{i=1}^N m_{t,i}\bar g_t(x_i)
\to
\sum_{i=1}^N \bar g_t(x_i).
\]
Substituting this limit into the recursive query gives the full-participation
counterpart.
\end{proof}



\newpage
\section*{NeurIPS Paper Checklist}

\begin{enumerate}

\item {\bf Claims}
    \item[] Question: Do the main claims made in the abstract and introduction accurately reflect the paper's contributions and scope?
    \item[] Answer: \answerYes{}.
    \item[] Justification: The main claims in the abstract and introduction are restricted to the proposed FO-DP-SGD mechanism, its mechanism-level fractional-memory design, its compatibility with standard privacy accounting, and the evaluated privacy--utility improvements. These claims are supported by the methodology in Section~\ref{sec:fodpsgd_methodology}, the empirical results in Section~\ref{sec:experiment}, and the additional ablation studies in Appendix~\ref{app:additional-results}.

\item {\bf Limitations}
    \item[] Question: Does the paper discuss the limitations of the work performed by the authors?
    \item[] Answer: \answerYes{}.
    \item[] Justification: The paper discusses limitations in Section~\ref{sec:conclusion}, including the use of benchmark image-classification datasets and compact neural architectures, dependence on privacy-accounting assumptions and hyperparameter choices, and the need for further evaluation on larger-scale architectures and deployment settings.

\item {\bf Theory assumptions and proofs}
    \item[] Question: For each theoretical result, does the paper provide the full set of assumptions and a complete (and correct) proof?
    \item[] Answer: \answerYes{}.
    \item[] Justification: The paper states the FO-DP-SGD mechanism in Section~\ref{sec:fodpsgd_methodology} and provides the full theoretical analysis in Appendix~\ref{app:theoretical-analysis}. The appendix includes assumptions, propositions, theorems, and proofs for history-conditioned query structure, sensitivity, privacy accounting, adaptive R\'enyi DP composition, signal--memory--noise decomposition, and limiting regimes.

\item {\bf Experimental result reproducibility}
    \item[] Question: Does the paper fully disclose all the information needed to reproduce the main experimental results of the paper to the extent that it affects the main claims and/or conclusions of the paper, regardless of whether the code and data are provided or not?
    \item[] Answer: \answerYes{}.
    \item[] Justification: Section~\ref{subsec:experimental-setup-reproducibility} specifies the datasets, fixed train/test subsets, preprocessing, model architecture, privacy parameters, optimizer settings, FO-DP-SGD hyperparameters, baselines, random seeds, and statistical reporting protocol. Additional empirical results and ablations are provided in Appendix~\ref{app:additional-results}.

\item {\bf Open access to data and code}
    \item[] Question: Does the paper provide open access to the data and code, with sufficient instructions to faithfully reproduce the main experimental results, as described in supplemental material?
    \item[] Answer: \answerNo{}.
    \item[] Justification: The paper uses publicly available benchmark datasets and provides detailed experimental information in Section~\ref{subsec:experimental-setup-reproducibility} and Appendix~\ref{app:additional-results}. The source code is not released at submission time; after publication, we plan to make the implementation and scripts available upon reasonable request to support reproduction of the reported results.

\item {\bf Experimental setting/details}
    \item[] Question: Does the paper specify all the training and test details, e.g., data splits, hyperparameters, how they were chosen, type of optimizer, necessary to understand the results?
    \item[] Answer: \answerYes{}.
    \item[] Justification: Section~\ref{subsec:experimental-setup-reproducibility} specifies the training and test setup, including datasets, fixed train/test subsets, preprocessing, architecture, clipping norm, noise multiplier, sampling rate, privacy parameter, privacy accountant, learning rates, baseline optimizers, FO-DP-SGD hyperparameters, ablation ranges, and seed-level statistical reporting.

\item {\bf Experiment statistical significance}
    \item[] Question: Does the paper report error bars suitably and correctly defined or other appropriate information about the statistical significance of the experiments?
    \item[] Answer: \answerYes{}.
    \item[] Justification: The main comparison experiments are repeated over five independent random seeds. Section~\ref{subsec:experimental-setup-reproducibility} defines the statistical reporting protocol, and Appendix~\ref{subsec:final-accuracy-analysis} reports final accuracy as mean $\pm$ standard deviation together with 95\% confidence intervals computed using Student's $t$ interval.

\item {\bf Experiments compute resources}
    \item[] Question: For each experiment, does the paper provide sufficient information on the computer resources, type of compute workers, memory, time of execution, needed to reproduce the experiments?
    \item[] Answer: \answerYes{}.
    \item[] Justification: The paper reports runtime comparisons and relative runtime overhead in Appendix~\ref{subsec:runtime_comparison_cifar10} and Appendix~\ref{subsec:relative_runtime_overhead_cifar10}. The experimental setup in Section~\ref{subsec:experimental-setup-reproducibility} also records per-run runtime as part of the seed-level logs used to generate the reported summaries.

\item {\bf Code of ethics}
    \item[] Question: Does the research conducted in the paper conform, in every respect, with the NeurIPS Code of Ethics \url{https://neurips.cc/public/EthicsGuidelines}?
    \item[] Answer: \answerYes{}.
    \item[] Justification: The work is methodological and experimental, uses public benchmark datasets, and does not involve human subjects, private user data, deceptive systems, or high-risk deployment. Privacy-related assumptions, limitations, and implications are discussed in Section~\ref{sec:conclusion} and Section~\ref{sec:broader-impact}.

\item {\bf Broader impacts}
    \item[] Question: Does the paper discuss both potential positive societal impacts and negative societal impacts of the work performed?
    \item[] Answer: \answerYes{}.
    \item[] Justification: Section~\ref{sec:broader-impact} discusses positive impacts of improving utility in privacy-preserving learning, as well as possible negative impacts from incorrect privacy accounting, excessive privacy budgets, distribution shift, downstream misuse, fairness concerns, and residual information leakage.

\item {\bf Safeguards}
    \item[] Question: Does the paper describe safeguards that have been put in place for responsible release of data or models that have a high risk for misuse, e.g., pre-trained language models, image generators, or scraped datasets?
    \item[] Answer: \answerNA{}.
    \item[] Justification: The paper does not release high-risk pretrained models, language models, image generators, scraped datasets, or dual-use datasets. The experiments use standard public benchmark datasets and focus on a differentially private optimization method.

\item {\bf Licenses for existing assets}
    \item[] Question: Are the creators or original owners of assets, e.g., code, data, models, used in the paper, properly credited and are the license and terms of use explicitly mentioned and properly respected?
    \item[] Answer: \answerYes{}.
    \item[] Justification: The paper uses public benchmark datasets and standard software libraries, which are credited in the experimental setup and references. Dataset access and usage follow their public terms, and no proprietary or restricted datasets are introduced.

\item {\bf New assets}
    \item[] Question: Are new assets introduced in the paper well documented and is the documentation provided alongside the assets?
    \item[] Answer: \answerNA{}.
    \item[] Justification: The paper does not introduce or release a new dataset, benchmark, pretrained model, or standalone software asset at submission time. The contribution is a private optimization method, and the algorithmic and experimental details are documented in Section~\ref{sec:fodpsgd_methodology}, Section~\ref{subsec:experimental-setup-reproducibility}, and Appendix~\ref{app:theoretical-analysis}.

\item {\bf Crowdsourcing and research with human subjects}
    \item[] Question: For crowdsourcing experiments and research with human subjects, does the paper include the full text of instructions given to participants and screenshots, if applicable, as well as details about compensation, if any?
    \item[] Answer: \answerNA{}.
    \item[] Justification: The paper does not involve crowdsourcing, user studies, surveys, annotation tasks, or research with human subjects. All experiments are conducted on public machine-learning benchmark datasets.

\item {\bf Institutional review board (IRB) approvals or equivalent for research with human subjects}
    \item[] Question: Does the paper describe potential risks incurred by study participants, whether such risks were disclosed to the subjects, and whether Institutional Review Board approvals, or an equivalent approval/review based on the requirements of your country or institution, were obtained?
    \item[] Answer: \answerNA{}.
    \item[] Justification: The paper does not involve human subjects, crowdsourcing, private participant data, or user studies; therefore IRB approval or equivalent human-subjects review is not applicable.

\item {\bf Declaration of LLM usage}
    \item[] Question: Does the paper describe the usage of LLMs if it is an important, original, or non-standard component of the core methods in this research? Note that if the LLM is used only for writing, editing, or formatting purposes and does \emph{not} impact the core methodology, scientific rigor, or originality of the research, declaration is not required.
    \item[] Answer: \answerNA{}.
    \item[] Justification: LLMs are not used as an important, original, or non-standard component of the proposed method, theory, experiments, or evaluation. Any use of language tools, if applicable, is limited to writing, editing, or formatting assistance and does not affect the scientific content or core methodology.

\end{enumerate}

\end{document}